\documentclass[reprint,aps,pra,superscriptaddress]{revtex4-2}
\usepackage{amsmath,amsthm,amssymb,amscd,float,graphicx}
\usepackage[T1]{fontenc}
\usepackage{lmodern}
\usepackage[colorlinks,linkcolor=blue,citecolor=blue,urlcolor=blue]{hyperref}
\usepackage{lipsum}
\usepackage{subfig} 
\usepackage{algorithm}
\usepackage{algpseudocode}
\usepackage[normalem]{ulem}
\newcommand{\al}{\alpha}
\newcommand{\be}{\beta}
\newcommand{\de}{\delta}
\newcommand{\ep}{\epsilon}

\newcommand{\si}{\sigma}
\renewcommand{\th}{\theta}

\newcommand{\De}{\Delta}

\newcommand{\ba}{\mathbf a}
\newcommand{\bb}{\mathbf b}
\newcommand{\cA}{{\mathcal A}}
\newcommand{\cB}{{\mathcal B}}

\newcommand{\id}{1\hspace{-.25em}{\rm I}}

\newcommand{\ket}[1]{|#1\rangle}
\newcommand{\bra}[1]{\langle#1|}

\newcommand{\mss}{\kern 1pt}
\renewcommand{\leq}{\leqslant}
\renewcommand{\geq}{\geqslant}
\renewcommand{\le}{\leqslant}
\renewcommand{\ge}{\geqslant}
\newcommand{\tends}[1]{\bbuildrel{\hbox to 2em{\rightarrowfill}}_{#1}^{}}

\newcommand{\tr}{\operatorname{tr}}

\newcommand{\e}{\mathrm e}

\newcommand{\ie}{\textit{i.e. }}

\newcommand{\Int}[1]{\,\mathop{\!#1}\limits^{\lower1ex\hbox{$\scriptstyle\circ$}}{}}
\newtheorem{thm}{Theorem}
\newtheorem{defi}{Definition}
\newtheorem{prop}{Proposition}
\newtheorem{lemma}{Lemma}

\theoremstyle{definition}

\theoremstyle{remark}

\begin{document}
\title{Highly-entangled, highly-doped states that are efficiently cross-device verifiable}
\author{Janek Denzler}
\affiliation{Dahlem Center for Complex Quantum Systems, Freie Universit\"at Berlin, 14195 Berlin, Germany\looseness=-1}
\author{Santiago Varona}
\affiliation{Instituto de F\'isica Te\'orica, UAM-CSIC, Universidad Aut\'onoma de Madrid, Cantoblanco, 28049 Madrid, Spain}
\author{Tommaso Guaita}
\author{Jose Carrasco}
\email[Corresponding author:\,]{jose.carrasco@fu-berlin.de}
\affiliation{Dahlem Center for Complex Quantum Systems, Freie Universit\"at Berlin, 14195 Berlin, Germany\looseness=-1}
\date{\today}
\begin{abstract}
In this paper, we introduce a class of highly entangled real quantum states that cannot be approximated by circuits with $\log$-many non-Clifford gates and prove that Bell sampling enables efficient cross-device verification (or distributed inner product estimation) for these states. That is, two remote parties can estimate the inner product $\tr(\rho\si)$, each having black-box access to copies of a state $\rho$ (or respectively~$\si$) in this class. This is significant because it is clear that this task can be achieved in those cases (such as low entanglement or low non-Clifford gate count) where one can independently learn efficient classical descriptions of each state using established techniques and share the description to compute the overlap. Instead, our results demonstrate that this is possible even in more complex scenarios where these ``learn and share'' methods are insufficient. Our proposal is scalable, as it just requires a number of two-copy Bell measurements and single-copy Pauli measurements that grows polynomially with both the number of qubits and the desired inverse-error, and can be implemented in the near term. Moreover, the required number of samples can be efficiently experimentally determined by the parties in advance, and our findings are robust against preparation errors. We anticipate that these results could have applications in quantum cryptography and verification.
\end{abstract}
\maketitle
\section{Introduction}
As quantum computing advances, ensuring the reliability of implemented quantum states remains a key challenge~\cite{EHW20}. This includes verifying the similarity of states prepared across different devices. Standard benchmarking techniques, such as quantum state tomography~\cite{OW16} and direct fidelity estimation~\cite{DL11,FL11}, rely on first classically learning and sharing state descriptions to then assess the states' similarity. However, these “learn and share” methods are inherently 
limited to states that can be efficiently represented classically, such as those with low entanglement~\cite{SWVC08} or requiring only a few non-Clifford gates to prepare~\cite{CLL24,Leone2024learningtdoped}. At the same time, other solutions based on interactive proofs and post-quantum cryptography ({\em e.g.}~\cite{Ma18,BKVV20,BCMVV21}) require resources that are far from being reached in the near future.
 
In this work, we explore how two parties can assess the similarity of their quantum states efficiently in practically relevant scenarios, without requiring to reveal an efficient description of their structure (see~Fig.~\ref{fig:pic1}). This is also relevant outside the benchmarking field and for applications in a fault-tolerant setting. Indeed, it would allow secure quantum verification, where parties can confirm that they hold the same quantum state without revealing its content.

\begin{figure}[ht]
\includegraphics[width=.9\linewidth]{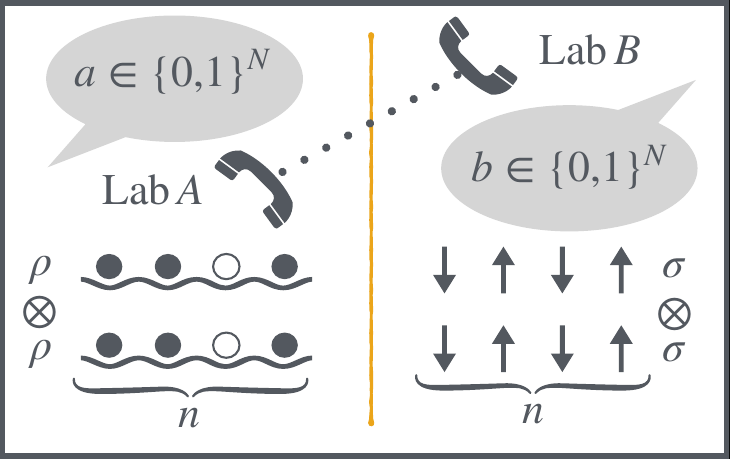}
\caption{\raggedright \small Two distant quantum devices potentially operating on different platforms want to estimate the inner product $\tr(\rho\sigma)$ using only LOCC. The parties exchange $N$ bits of information $a,b$, which allows to construct an estimator for the inner product. The protocol is efficient if it is sufficient to choose $N$ polynomial in the system size $n$. For low-entangled or low-doping states, $a,b$ can be just compact classical descriptions that allow the reconstruction of $\rho,\si$. In this work we ask: is it possible to do efficient cross-device verification for states without sharing classical descriptions, \ie with highly-entangled and highly-doped states?}\label{fig:pic1}
\end{figure}

Cross-device verification~\cite{El20,CEKKZ21,ALL22} requires estimating the overlap between states prepared on different platforms. The Hadamard test enables this via quantum communication, and recent works \cite{Kn23, Go24, Ar24} have also explored the prospect of inner product estimation with limited quantum communication. While this is a valid strategy, it is often impractical, particularly for devices based on different technologies or located at distant places. Instead, doing this just by local measurements combined with classical communication (LOCC) leads to the problem of distributed inner product estimation (DIPE). Existing DIPE methods, such as those based on randomized measurements and classical shadows, avoid explicit state learning but face exponential complexity for generic states~\cite{El20,NN24}.

A promising alternative is cross-device verification via Pauli sampling~\cite{HIJLEC24}, which also estimates  this overlap without learning state descriptions. Crucially, Pauli sampling tailors the measurement schedule to the input states, making it efficient for certain families of states, depending on properties of their Pauli distribution. We demonstrate that via Pauli sampling, DIPE is efficient in cases where traditional ``learn and share'' strategies fail. Specifically, we introduce a class of highly entangled, highly doped real quantum states that are inaccessible to classical learning techniques. Our findings show that LOCC-based DIPE is feasible even when tomography is impractical, enabling verification in challenging quantum computing scenarios. Therefore, our approach is a step towards confidential and efficient verification, allowing parties to estimate state overlaps without revealing enough information for reconstruction. This is a crucial advancement both for benchmarking applications and for the development of secure quantum communication and cryptographic protocols.

One limitation is that our method requires states to be real. However, we show that it remains effective under realistic noise conditions: if a state is close to a real target pure state in trace norm, the protocol still yields accurate overlap estimates. This robustness ensures practicality despite experimental imperfections, such as small non-real gate errors. Additionally, for near-pure real states, remote parties can predict the number of samples needed to achieve a desired accuracy, enabling efficient and reliable verification. In conclusion, our work bridges the gap between efficiency and confidentiality in cross-device verification, providing a scalable solution for efficient and secure quantum state validation in both near-term and fault-tolerant quantum technologies.

\section{The ${\rm rDIPE}$ algorithm and its complexity} 
In Ref.~\cite{HIJLEC24}, an algorithm for the distributed estimation of the inner product $\tr(\rho\sigma)$ between two pure quantum states was introduced, which only relies on performing one and two copy measurements of either state. We use a version of this algorithm that also works for mixed states, assuming that $\tr\rho^2,\tr\sigma^2 >1/2$. For completeness we provide it in Appendix~\ref{app:rDIPE} (see Algorithm~\ref{alg:rDIPE}),

In order to properly discuss the protocol's efficiency, let us introduce some notation. 
Given a quantum state $\rho$, we denote its Pauli distribution by $p_\rho$ and its Bell distribution by $q_\rho$. That is,
\[
    p_\rho(\ba)=\frac{\langle P_\ba\rangle_\rho^2}{2^n\tr\rho^2}\,,  \hspace{10mm}
    q_\rho(\ba)=\bra{\Phi_\ba}\rho\otimes\rho\ket{\Phi_\ba}\,,
\]
for $\ba\in\{0,1,2,3\}^n$, where $P_{\ba}=P_{a_1}\otimes \cdots \otimes P_{a_n}$ are $n$ qubit Pauli operators and $\ket{\Phi_{\ba}}=(\id\otimes P_{\ba})\ket{\Phi_{\mathbf{0}}}$ are the Bell states, with $\ket{\Phi_{\mathbf{0}}}=1/\sqrt{2^n} \sum_{z\in\{0,1\}^n}\ket{z_1,\dots,z_n}\otimes \ket{z_1,\dots,z_n}$. Notice that if $\rho$ is pure and has real amplitudes then $p_{\rho}=q_{\rho}$~\cite{Mo17}. Then, we denote by $F_\rho$ the cumulative distribution function (CDF) of the random variable $\langle P_{\ba}\rangle_\rho^2$ when $\ba$ is sampled according to $p_\rho$, that is
\begin{equation}
    F_\rho(\ep)=\sum_{\ba }p_\rho(\ba) \, \theta(\ep-\langle P_\ba\rangle_\rho^2) \,, \label{eq:CDF}
\end{equation}
where $\th(t)$ is the indicator function for the set $\{t\ge0\}$. Finally, we will consider sequences of states indexed by the number of qubits and denote these as $\{\rho\}\equiv\{\rho_n\}_n$. 

The protocol ${\rm rDIPE}$, presented in detail in Appendix~\ref{app:rDIPE}, outputs an approximation $f(\rho,\si,N_1,N_2)$ of the quantity
\begin{equation}\label{eq:cos}
{\rm c}(\rho,\si):=\tr(\rho\si)/\sqrt{\tr\rho^2\tr\si^2}\,,
\end{equation}
computed by using using $N_1$ copies of $\rho^{\otimes 2}$ and $\si^{\otimes 2}$ and $N_1N_2$ copies of $\rho$ and $\si$ (see Algorithm~\ref{alg:rDIPE} in Appendix~\ref{app:rDIPE} for a precise definition of this estimator $f$).
In other words, it evaluates the cosine of the angle between $\rho$ and $\si$, as determined by the inner product $\tr(\rho\si)$. We say that ${\rm rDIPE}$ is efficient on inputs $\{\rho\},\{\si\}$ if the following holds~\footnote{Here and in what follows, ${\mathbb P}$ will denote the probability with respect to the internal randomness of rDIPE for fixed inputs.}:

\begin{defi}[Efficiency of ${\rm rDIPE}$]{\ \\}
\label{def:efficiency}
For two sequences $\{\rho\},\{\si\}$, we say that ${\rm rDIPE}$ is efficient on inputs $\{\rho\},\{\si\}$ if there exist integers $N_1,N_2={\rm poly}(n,1/\ep,\de)$
such that for all $n,\ep,\de>0$
\begin{equation}\label{eq:efficient}
    {\mathbb P}\Big\{\left|f(\rho_n,\si_n,N_1,N_2)-{\rm c}(\rho_n,\si_n)\right|>\ep\Big\}<\e^{-\de}.
\end{equation}
\end{defi}

In order to make statements about the efficiency of rDIPE on specific states, we will use a technical result proven in Ref.~\cite{HIJLEC24} which expresses how well the quantity ${\rm c}(\rho,\si)$ of Eq.~\eqref{eq:cos} can be approximated by the output $f(\rho,\si,N_1,N_2)$ of the protocol, 
by relating this to properties of the states' CDFs and Bell distributions.
The proof of this statement applies with minor modifications also to our case of mixed states, showing that for any $\ep_1>0$ and $\ep_2>0$
\begin{multline}\label{eq:performance}
{\mathbb P}\Big\{\left|f(\rho,\si,N_1,N_2)-{\rm c}(\rho,\si)\right|>4\ep_1+4\sqrt{\ep_2}+2F(\ep_2)+6\De\Big\}\\
<4\exp\left(-2\epsilon_1^2 N_1\right)+4N_1\exp\left(-\epsilon_2^2N_2/2\right),
\end{multline}
where $F=(F_\rho+F_\si)/2$ and $\De$ expresses the total variation (TV) distance between the Pauli distribution $p_{\rm mix}=(p_\rho+p_\si)/2$ and the Bell distribution $q_{\rm mix}=(q_\rho+q_\si)/2$, that is $\De={\rm TV}(p_{\rm mix},q_{\rm mix})$.

\section{The ${\rm CW}$ family of states} 
We now present our main results. That is, we introduce the following class of states, for which rDIPE is efficient, that are however not efficiently learnable with MPS or stabilizer-based methods. In the following, we will denote by ${\rm rCl}(n)$ the real Clifford group on $n$ qubits (see Ref.~\cite{HFGW18} for a detailed definition and summary of main properties).

\begin{defi}\label{def:CW}
    Let $n>0$ denote the number of qubits. We define the class ${\rm CW}(n)$ as the family of states of the form $C\ket{W_n}$, where $C$ is a real Clifford unitary on $n$ qubits, and $\ket{W_n}$ is the $n$-qubit $W$-state. Specifically,
    \[
    {\rm CW}(n)=\{\ket\psi:\ket\psi=C\ket{W_n}\,,\,C\in{\rm rCl}(n)\},
    \]
    where $\ket{W_n}=\sum_{|z|=1}c_z\ket{z_1,\ldots,z_n}$ with $c_z=1/\sqrt n$, and the sum extends over all $z\in\{0,1\}^n$ such that $\sum_iz_i=1$.
\end{defi}

   It is worth noting that our results are also valid if we replace the W-state of the previous definition by a Dicke state with a constant number of excitations, or by any other non-permutationally invariant superposition of computational basis states with constant Hamming weight like, {\em e.g.}, those of Ref.~\cite{RSPL24}. More generally, the same is true for states of the form $C\sum_{z\in S}c_z\ket{z_1,\ldots,z_n}$, where $C$ is a (highly-entangling) real Clifford gate, and we take, for instance, $c_z\propto\pm1/\sqrt{{\rm poly}(n)}$ for possibly different polynomials and the sum $z\in S$ running over any $S\subset\{0,1\}^n$ with cardinality $|S|={\rm poly}(n)$. However, for the purposes of this work, we shall favor the simplicity of Definition~\ref{def:CW}.

In the following, when dealing with sequences of states we will write $\{\rho\}\in{\rm CW}$ if $\rho_n\in{\rm CW}(n)$ for all $n$. First of all, we observe that ${\rm rDIPE}$ is efficient on inputs in CW.
\begin{prop}\label{prop:efficiency}
    For all inputs $\{\rho\},\{\si\}\in{\rm CW}$, $\mathrm{rDIPE}$ is efficient in the sense of Definition~\ref{def:efficiency}.
\end{prop}
\begin{proof}
To see this, for all $n$, $\ep>0$ consider $\ep_1,\ep_2$ given by $\ep_1=\ep/8$ and $\ep_2=\min\{(\ep/8)^2,3/n^2\}$. The main property of the states in ${\rm CW}$ that we will use here is that their CDF $F(\epsilon)$ vanishes when $\epsilon$ is inverse-polynomially small. This is proven in Appendix~\ref{app:CDFs} (see specifically Eq.~\eqref{eq:CDF-W}). In particular, this fact allows us to conclude that $F_{\rho_n}(\ep_2)=F_{\si_n}(\ep_2)=0$, since $\ep_2\le3/n^2<4/n^2$.

It follows that $4\ep_1+4\sqrt{\ep_2}+2F(\ep_2)\leq\ep$, which in turn implies
\begin{align}
    &{\mathbb P}\Big\{\left|f(\rho_n,\si_n,N_1,N_2)-{\rm c}(\rho_n,\si_n)\right|>\ep\Big\} \nonumber\\
    &\hspace{0mm}\leq{\mathbb P}\Big\{\!\!\left|f(\rho_n,\si_n,N_1,N_2)-{\rm c}(\rho_n,\si_n)\right|\!>\!4\ep_1+\!4\sqrt{\ep_2}+\!2F(\ep_2)\!\Big\} \nonumber\\
    &\hspace{4mm}<4\exp\left(-2\epsilon_1^2N_1\right)+4N_1\exp\left(-\epsilon_2^2N_2/2\right), \label{eq:sets-relation_efficiency}
\end{align}
where in the last step we have used Eq.~\eqref{eq:performance}, considering that $\rho_n,\si_n$ have real amplitudes which implies $\Delta=0$. This shows that there are $N_1,N_2={\rm poly}(n,1/\ep,\de)$ such that Eq.~\eqref{eq:efficient} holds for $\{\rho\},\{\si\}\in{\rm CW}$.
\end{proof}

Next, we show that the class ${\rm CW}$ contains highly entangled states that cannot be efficiently approximated by matrix product states (MPS). This is summarised by the following proposition:

\begin{prop}\label{prop:high-ent}
    There exists at least one family of states $\{\rho\}\in {\rm CW}$ such that 
    \begin{equation}
        S_2(\tr_{n/2}\rho_n)\geq c \, n\,, \label{eq:high-entanglement_statement}
    \end{equation}
    for some constant $c$. Here, $\tr_{n/2}\rho_n$ is the reduced state obtained by tracing out the first half of the system's qubits and $S_2$ is the second Rényi entropy. According to well-known results (see for example Ref.~\cite{SWVC08}), this implies that $\{\rho\}$ cannot be approximated to arbitrary precision by a family of MPS with bond dimension polynomial in~$n$.
\end{prop}

\begin{proof}
    To prove this statement, we first show that the average entanglement entropy of the states in the class ${\rm CW}(n)$ grows linearly in $n$, in particular
    \begin{equation}
        \mathbb{E}[S_2(\tr_{n/2}C\ket{W_n}\bra{W_n}C^\dag)]\geq c \,n \,. \label{eq:average-entanglement}
    \end{equation}
    Here, $\mathbb{E}$ denotes the average over the unitary $C$ according to the Haar measure of the real Clifford group ${\rm rCl}(n)$.

    To see this, notice that using standard techniques we can rewrite the left hand side of the expression above as
    \begin{align}
        &\mathbb{E}\left[ -\log\,\tr\left(\mathrm{SWAP}_{n/2} \, \left( C\ket{W_n}\bra{W_n}C^{\dag}\right)^{\otimes 2} \right) \right]  \nonumber\\
        &\geq-\log\tr \Big[ \mathrm{SWAP}_{n/2} \; \mathbb{E} \!\left( C^{\otimes 2}(\ket{W_n}\bra{W_n})^{\otimes 2} (C^\dag)^{\otimes 2}\right)\Big]  \nonumber\\
        &\geq -\log\tr \Big[ \mathrm{SWAP}_{n/2} \; \frac{1}{2^{2n}}(k \, P_{\rm sym}+ 2^nk'\ket{\Phi_{\bf 0}}\!\bra{\Phi_{\bf 0}}) \Big] , \label{eq:Haar-integral-result}
    \end{align}
    where $\mathrm{SWAP}_{n/2}$ is the operator that swaps the first $n/2$ qubits of the first factor of the tensor product with the  first $n/2$ qubits of the second factor. In the last step we have used Haar integration techniques to evaluate the average value of the quantity involving $C^{\otimes 2}$, using that the real Clifford group is an orthogonal $2$-design~\cite{HFGW18}. 
    More details on this calculation are presented in Appendix~\ref{app:Haar} (see Theorem~\ref{thm:comm}, Lemma~\ref{lemma:comm} and following discussion). In particular, we have that $P_{\rm sym}=(\id+\mathrm{SWAP}_{n})/2$ and $k$, $k'$ are constants. Here, $\mathrm{SWAP}_{n}$ swaps the entire first factor of the tensor product with the second. By straightforward computation we find $\tr(\mathrm{SWAP}_{n/2}P_{\rm sym})=2^{3n/2}$ and $\tr(\mathrm{SWAP}_{n/2}\ket{\Phi_{\bf 0}}\bra{\Phi_{\bf 0}})=1$, which implies that~\eqref{eq:Haar-integral-result} is equal to $-\log(k\, 2^{-n/2}+k'\, 2^{-n})\geq c \, n$.
    
    Since, according to~\eqref{eq:average-entanglement}, the states $\mathrm{CW}(n)$ on average have an entanglement entropy of at least $c\,n$, there must exist at least one state in $\mathrm{CW}(n)$ whose entanglement entropy is larger or equal to this average value. By setting $\rho_n$ equal to such state for all $n$, we construct a family $\{\rho\}$ that satisfies~\eqref{eq:high-entanglement_statement}.
\end{proof}

\begin{figure}[t]
\includegraphics[width=.99\linewidth]{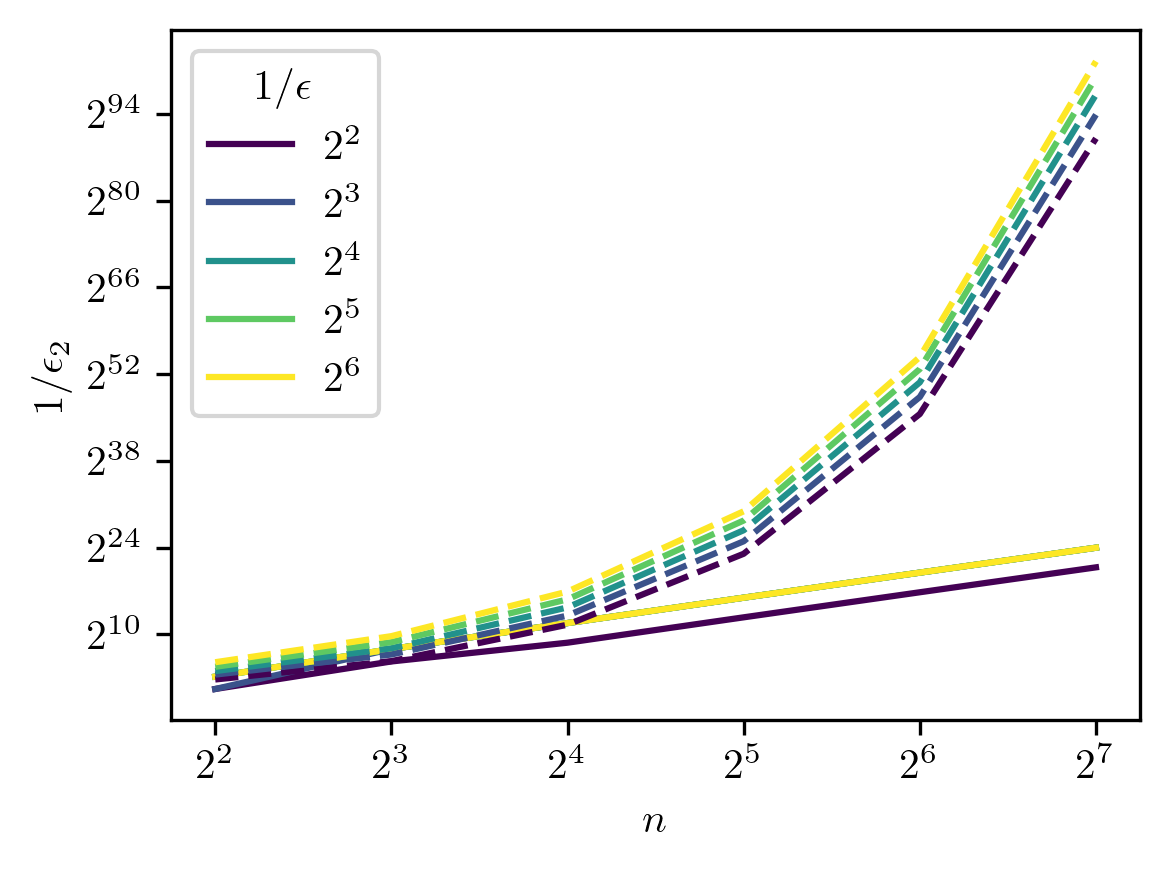}
\caption{\raggedright \small Here we plot $1/\ep_2$ in terms of $n$ for different values of $1/\ep$ and for two families of states where the protocol is efficient (solid) and inefficient (dashed). To illustrate the measurement protocol we consider states that have an efficient MPS representation and thus can be simulated (here we use TenPY~\cite{TenPy}), however the same behaviour of the solid lines is also expected for other families of states with extensive entanglement entropy (see Definition~\ref{def:CW}, the remark afterwards and Proposition~\ref{prop:high-ent}). Precisely, we compute $\ep_2$ as the solution of $F_{\rho,N,K}(\ep_2)=\epsilon$ (see discussion after Eq.~\eqref{eq:FNK}). The solid lines correspond to the Dicke state with two excitations, whereas the dashed ones correspond to a translationally invariant random MPS of bond dimension $2$ where the entries of each tensor are iid according to the standard normal distribution. We take $N=5\times10^4$ and $K$ effectively infinite (as for MPS $\langle P\rangle^2_\rho$ can be computed exactly).}\label{fig:scaling_eps2_eps_n}
\end{figure}

A final property of the states in ${\rm CW}(n)$ is that they cannot be approximated by quantum circuits with just few non-Clifford gates. Precisely, states prepared by Clifford circuits supplemented by at most $t$ T-gates (known in the literature as \emph{$t$-doped states}~\cite{CLL24,Leone2024learningtdoped}) can only be good approximations to states in ${\rm CW}(n)$ if $t$ grows linearly in~$n$. This is summarized in the following proposition~\footnote{The proof is analogous for arbitrary non-Clifford gates acting on a constant number of qubits}:

\begin{prop}\label{prop:high-t}
    Let $\rho'$ be a $t$-doped state on $n$ qubits. If $\|\rho-\rho'\|_{\rm tr}<1/4$ for some $\rho\in{\rm CW}(n)$, then $t>n/4$. 
\end{prop}

\begin{proof}
     We show the contrapositive: $t$-doped states with $t\leq n/4$ are at least $1/4$ far in trace distance from ${\rm CW}(n)$. A simple counting argument suffices: on the one hand, arbitrary $t$-doped states $\rho'$ on $n$ qubits have at least $2^{n-t}$ stabilizers, as the application of a single $T$-gate reduces the number of stabilizer generators by at most one~\cite{CLL24}. On the other hand, we can upper-bound the number of Paulis with high expectation value for states $\rho$ in ${\rm CW}(n)$, as we discuss in Appendix~\ref{app:CDFs}. Concretely, for $t \leq n/4$, we find that
     \[
     \begin{aligned}
         |\{\ba:|\langle P_\ba\rangle_\rho|>3/4\}|&<2^{3n/4}\\
         |\{\ba:|\langle P_\ba\rangle_{\rho'}|=1\}|&\ge2^{n-t}\ge 2^{3n/4},
     \end{aligned}
     \]
     where we have used bound~\eqref{eq:counting-Ps} in Appendix~\ref{app:CDFs} to obtain the first inequality. Therefore, there exists at least one Pauli operator $P$ with $|\langle P\rangle_\rho-\langle P\rangle_{\rho'}|\ge1/4$, and thus $\|\rho-\rho'\|_{\rm tr}\ge1/4$, as claimed.
\end{proof}

We have constructed a (generalizable) class of states denoted ${\rm CW}$, see Definition~\ref{def:CW}. Proposition~$\ref{prop:high-ent}$ shows that this class contains sequences of highly-entangled states, for which MPS learning algorithms are known to fail~\cite{SWVC08}. In fact, as a consequence of Markov's inequality, a large fraction of sequences will have this property. Furthermore, due to Proposition~$\ref{prop:high-t}$, the states in ${\rm CW}$ cannot be well-approximated by $(\log n)$-doped states, for which efficient learning algorithms also exist~\cite{CLL24,Leone2024learningtdoped}. Nonetheless, Algorithm~\ref{alg:rDIPE} is efficient on states in ${\rm CW}$, as specified by Proposition~\ref{prop:efficiency}.

\section{Robustness} 
In any realistic experimental setting, state preparation is subject to noise. Even if the protocol ${\rm rDIPE}$ is efficient on inputs $\{\rho\}, \{\si \}$, it is a priori unclear how the protocol will behave on noisy inputs $\{\rho'\}, \{\si'\}$. In particular, we want that small state preparation errors still allow for efficiently estimating $\tr(\rho'\si')$, rendering the protocol practically meaningful. This is complicated by the fact that general noise channels will not preserve the realness of our target states, even if these are ideally prepared by real gates only~\cite{aharonov2003simple, shi2002toffoli}. Note here that an arbitrary noise channel can always be mapped to a realness-preserving Pauli channel by randomized compilation~\cite{WE16}, even though our proposal does not rely on this technique. However, independently of this, in Appendix~\ref{app:robustness} we show that ${\rm rDIPE}$ is robust on inputs $\{\rho\},\{\si\}\in{\rm CW}$, in the following precise sense.

Given two sequences $\{\rho\},\{\si\}\in{\rm CW}$, there exists $k>0$ and integers $N_1,N_2={\rm poly}(n,1/\tau,\de)$
such that for all $n,\tau,\de>0$
\begin{equation}\label{eq:efficient-robust}
    {\mathbb P}\Big\{\left|f(\rho'_n,\si'_n,N_1,N_2)-{\rm c}(\rho'_n,\si'_n)\right|>k\tau\Big\}<\e^{-\de}\,,
\end{equation}
for all other sequences $\{\rho'\},\{\si'\}$ such that $\sup_n\{\|\rho_n-\rho'_n\|_{\rm tr}\} \leq\tau$ and analogously for $\si$.
That is, even if we run ${\rm rDIPE}$ on not necessarily real inputs $\{\rho'\}, \{\si'\}$ that are within trace-norm $\tau$ of the real pure states $\{\rho\}, \{\si\}\in {\rm CW}$, where the protocol is efficient, we are still able to estimate the quantity ${\rm c}(\rho',\si')$ of Eq.~\eqref{eq:cos}  with polynomial resources, but up to an error linear in $\tau$.

To round off this section on the robustness of our protocol, let us note that the practical impossibility of preparing two exact copies $\rho\otimes\rho$ or $\sigma\otimes\sigma$ for Bell sampling can be treated in an analogous way. Indeed, any possibly entangled state $\tilde\rho$ with $\|\tilde\rho-\rho\otimes\rho\|<\tau$ allows for sampling from $q_\rho$, up to an error at most $\tau$ in TV-distance.

\section{Estimation of required resources} 
A further appealing aspect of the presented protocol is that, given arbitrary real states (not necessarily known to be efficient), the remote parties can experimentally determine the number of samples required for distributed inner product estimation as follows: First, collect $N$ samples $\{\ba_1,\ldots,\ba_{N}\}$ from the Pauli distribution of the $n$-qubit pure state $\rho$. This can be achieved by measuring in the Bell basis, as for real states $q_\rho(\ba)=p_\rho(\ba)$. Then for each sample construct an estimate of $\langle P_{\ba_i}\rangle_\rho^2$ using $K$ further measurements. Let us call this estimate $\al_i^2(K)$. Consider now the previously introduced CDF~\eqref{eq:CDF} and its empirical estimator
\begin{equation}
    F_{\rho,N,K}(x)=\frac{1}{N}\sum_{i=1}^N\th\big(x-\al_i^2(K)\big)\,. \label{eq:FNK}
\end{equation}
Using the Dvoretzky–Kiefer–Wolfowitz~\cite{dvoretzky1956asymptotic,DKW} and Hoeffding inequalities, one can show that there are $N={\rm poly}(\de,1/\ep)$ and $K={\rm poly}(\de,1/x,\log(1/\ep))$ such that, if $F_{\rho,N,K}(2x)\le\ep/2$ then
\[
{\mathbb P}\big(F_\rho(x)>\ep\big)<\e^{-\de}\,.
\]
This means that, by looking at $F_{\rho,N,K}$ with this choice of $N$, $K$, we can find a sufficiently small $\ep_2$ such that with high probability $2F_\rho(\ep_2)+4\sqrt{\ep_2}\le\ep/2$. Consider now $\ep_2$ as a function of $n,1/\ep$. Then, similarly to Eq.~\eqref{eq:sets-relation_efficiency}, it follows that if $1/\ep_2={\rm poly}(n,1/\ep)$ the protocol is efficient. One can see that this is equivalent to the solution of $F_{\rho,N,K}(x)=\epsilon$ being $x=1/{\rm poly}(n,1/\ep)$. We illustrate this approach in Fig.~\ref{fig:scaling_eps2_eps_n} for two different families of states.

\section{Conclusions} In our work, we construct explicit examples of states without compact classical description (in the form of MPS or $t$-doped states) for which the distributed estimation of inner products is efficient via Pauli sampling. We would like to finish by commenting on the practical utility of our work and its limitations.

For our proposal to be implemented, each party needs to be able to approximately prepare a two-copy pure (and real) quantum state, a theoretical primitive that has been shown to be of interest in different contexts (see, e.g.~\cite{Science-two-copies-theory} and references therein). We stress that each such two-copy state has to be prepared within the same lab, something that has actually been demonstrated in practice in different settings like photonic systems~\cite{two-copies-photons}, trapped-ions~\cite{two-copies-trapped-ions} or atomic arrays~\cite{two-copies-ultra-cold-bosonic-atoms,two-copies-lukin}, and even for logical quantum states~\cite{two-logical-copies-lukin}. Nonetheless, it is admittedly true that being able to prepare such a high-purity, two-copy state is experimentally challenging and already indicates a high quality of the quantum devices involved. 
However, we think that our proposal can still be practically useful in scenarios where the parties use different quantum implementations, or they cannot natively implement the same gate-sets. For instance, our proposal can be used to efficiently check whether two sequences of different apparatus do indeed produce similar wave-functions (even if an efficient description of the underlying wave-function is not available).

Our results further show that the amount of classical information that needs to be exchanged to estimate the inner product between certain quantum states via LOCC can be smaller than the amount of information needed to characterise the quantum states themselves. While our work primarily serves as a proof of concept, we hope that further investigation may generalise these findings to a broader class of states, potentially unlocking new applications in quantum cryptography and verification.

\section*{Acknowledgments} 
We would like to thank R. Ruiz for pointing out Ref.~\cite{RSPL24} and useful discussions. TG was supported by the German Federal Ministry for Education and Research (BMBF) under the project FermiQP. JC acknowledges financial support from BerlinQuantum, BMBF (DAQC, MUNIQC-Atoms), and the project Munich Quantum Valley (K8). JD acknowledges financial support from BMBF (DAQC). SV acknowledges support from PID2021-127726NB- I00 (MCIU/AEI/FEDER, UE), from the Grant IFT Centro de Excelencia Severo Ochoa CEX2020-001007-S, funded by MCIN/AEI/10.13039/501100011033, and from the CSIC Research Platform on Quantum Technologies PTI-001.

\clearpage

\appendix
\section{The {\rm rDIPE} algorithm}\label{app:rDIPE} 
The protocol for real distributed inner product estimation 
(rDIPE) which we analyze in this work is a particular instance of the so-called symmetric protocol of Ref.~\cite{HIJLEC24}, adapted to mixed states with $\tr\rho^2,\tr\sigma^2>1/2$. It can be summarized as follows: 

\begin{algorithm}[H]
    \caption{rDIPE}\label{alg:rDIPE}
    \raggedright\textbf{Input:} $N_1$ copies of the unknown states $\rho\otimes\rho$ and $\si\otimes\si$, and $N_1N_2$ copies of the states $\rho$ and $\sigma$.\\
    \textbf{Output:} $f(\rho,\si,N_1,N_2)$, an estimate of $c(\rho,\si)$
    \begin{algorithmic}[1]
    \State  Alice and Bob sample $N_1$ times from the distribution $q_{\rm mix}=(q_\rho+q_\si)/2$ to obtain $\{\ba_1,\ldots,\ba_{N_1}\}$. That is, each $\ba_i$ is the outcome of a Bell measurement in either $\rho\otimes\rho$ or $\si\otimes\si$ with equal probability.
    \ForAll{$i\in\{1,2,\ldots,N_1\}$}
        \State {\bf Alice}: using $N_2$ single-copy measurements, construct $\al_i(N_2)$, an estimate of the expectation value $\langle P_{\ba_i}\rangle_\rho$.
        \State {\bf Bob}: using $N_2$ single-copy measurements, construct $\be_i(N_2)$, an estimate of the expectation value $\langle P_{\ba_i}\rangle_\si$.
    \EndFor
    \State \textbf{Return}
    \begin{multline*}
    f(\rho,\si,N_1,N_2)=\\
    \frac{1}{N_1}\sum_{i=1}^{N_1}\frac{2\,[\al_i(N_2)]\,[\be_i(N_2)]}{[\al_i(N_2)]^2\sqrt{B/A}+[\be_i(N_2)]^2\sqrt{A/B}}\,,
    \end{multline*}
    where $A$ and $B$ are estimates of $\tr\rho^2$ and $\tr\si^2$.
    \end{algorithmic}
\end{algorithm}

In this work, to simplify derivations, we assume that $A$ and $B$ are exactly equal to the purities $\tr\rho^2$, $\tr\si^2$. In practice, one would need to estimate them using $N_3$ additional measurement shots. This will lead to a further error term on the final inner product estimation, which can however be bounded with similar techniques to the other terms and can be made arbitrarily small by choosing a suitable $N_3$.

\section{On Haar integration over the real Clifford subgroup}\label{app:Haar}
Here we provide the details omitted in the proof of Proposition~\ref{prop:high-ent}. To evaluate twirls with respect to the real Clifford subgroup, we make use of the following well-known result (see Ref.~\cite{mele2024introduction} and references therein for a recent review on related techniques):

\begin{thm}[Commutant theorem]\label{thm:comm}
    Let $G$ be a compact group with Haar measure $\mu_G$ and $\{U(g)\}_{g \in G}$ a (finite-dimensional) unitary representation of $G$ on $(\mathcal{H}, \langle\cdot,\cdot\rangle)$. The $k$-fold twirl 
    \begin{equation*}
    \Phi_G^k: X \mapsto \int_G U(g)^{\otimes k} X \, U(g)^{\dag \otimes k} \, d\mu_G(g)
    \end{equation*}
    is an orthogonal projector (with respect to the Hilbert-Schmidt inner product) onto the commutant of $\{U(g)^{\otimes k}\}_{g\in G}$.
\end{thm}

As the real Clifford group rCl is an orthogonal 2-design~\cite{HFGW18}, the expectation value appearing in Eqs.~\eqref{eq:average-entanglement} and~\eqref{eq:Haar-integral-result} in the proof of Proposition~\ref{prop:high-ent} coincides with a $2$-fold twirl over the orthogonal group $G=O(2^n)$. To use the previous Theorem, we need the following characterization of the commutant of $O(2^n)\otimes O(2^n)$ (see Ref.~\cite{CS06} for details):

\begin{lemma} \label{lemma:comm}
    An orthonormal basis of the commutant of $O(2^n)\otimes O(2^n)$, with respect to the Hilbert-Schmidt inner product, is given by 
    \begin{align*}
        &\frac{1}{d_{\rm sym}^{1/2}}P_{\rm sym}  = \frac{1}{2\,d_{\rm sym}^{1/2}}(\id+{\rm SWAP}_n) \\[0.8em]
         &\frac{1}{d_{\rm asym}^{1/2}}P_{\rm asym} = \frac{1}{2\, d_{\rm asym}^{1/2}}(\id-{\rm SWAP}_n)\\[0.8em]
         & \overline{B} = \frac{1}{(1-1/d_{\rm sym})^{1/2}}\left(\ket{\Phi_{\bf 0}}\bra{\Phi_{\bf 0}}-\frac{1}{d_{\rm sym}}P_{\rm sym} \right),
    \end{align*}
    where $P_{\rm sym}$ and $P_{\rm asym}$ are projectors onto the symmetric and anti-symmetric subspaces with dimensions $d_{\rm sym}$ and $d_{\rm asym}$, respectively.
\end{lemma}

To complete the chain of inequalities~\eqref{eq:Haar-integral-result} in the proof of Proposition~\ref{prop:high-ent}, we project $\ket{W_n}\bra{W_n}^{\otimes 2}$ onto each orthogonal component separately. Due to swap-invariance, the projection onto $P_{\rm asym}$ vanishes, leaving us with a linear combination of $P_{\rm sym}$ and $\ket{\Phi_{\bf 0}}\bra{\Phi_{\bf 0}}$. Note that $d_{\rm sym} = 2^n(2^n+1)/2$. 

\section{CDF of states in {\rm CW}}\label{app:CDFs} 
In this section, we prove several properties of the CDF of states in ${\rm CW}$, which we use throughout the letter. They allow us to complete the proof of Proposition~\ref{prop:high-t} and analyse the efficiency and robustness of our protocol. We start with two Lemmas on the continuity of Pauli distributions and the corresponding CDFs that will be needed for the following discussion.

\begin{lemma}\label{lemma:tv0tv}
    Let $\rho,\si$ be quantum states with Pauli distributions $p_\rho,p_\si$. Then,
    \begin{align}
           \sum_{\ba}\left|(\tr\rho^2)\,p_\rho(\ba)-(\tr\si^2)\,p_\si(\ba)\right|&\le2\|\rho-\si\|_{\rm tr}\,,\label{eq:tv0}\\
           \max\{\tr\rho^2,\tr\si^2\}\, {\rm TV}(p_\rho,p_\si)&\le 2\|\rho-\si\|_{\rm tr}\,.\label{eq:tv}
    \end{align}
\end{lemma}

\begin{proof}
   We start by proving Eq.~\eqref{eq:tv0}. To this end, we will use the identity $\sum_{\ba}\langle P_{\ba}\rangle_A^2=2^n\tr A^2$, where $A$ is any Hermitian $n$-qubit operator. We upper-bound
    \[
    \begin{aligned}
    \sum_{\ba}\left|\langle P_{\ba}\rangle_\rho^2-\langle P_{\ba}\rangle_\si^2\right|&=\sum_{\ba}\left|\langle P_{\ba}\rangle_{\rho+\si}\right|\cdot\left|\langle P_{\ba}\rangle_{\rho-\si}\right|\\[0.5mm]
    &\le\sqrt{\sum_{\ba}\langle P_{\ba}\rangle_{\rho+\si}^2}\cdot\sqrt{\sum_{\bb}\langle P_{\bb}\rangle_{\rho-\si}^2}\\[1mm]
    &=2^n\sqrt{\tr(\rho+\si)^2}\cdot\sqrt{\tr(\rho-\si)^2}\\[3mm]
    &\le2\cdot2^n\, \|\rho-\si\|_{\rm tr} \, ,
    \end{aligned}
    \]
    where we have used the Cauchy-Schwarz and Hölder inequalities. Noting that $(\tr\rho^2)p_\rho(\ba)=\langle P_{\ba}\rangle_\rho^2/2^n$, this shows Eq.~\eqref{eq:tv0} from which follows Eq.~\eqref{eq:tv} after applying the triangle inequality. Indeed, assume without loss of generality that $\tr\rho^2\ge\tr\si^2$, and write
    \begin{align*}
           {\rm TV}(p_\rho,p_\si)\le  \hspace{2mm} & \frac12\sum_{\ba}\frac{\left|(\tr\rho^2)\,p_\rho(\ba)-(\tr\si^2)\,p_\si(\ba)\right|}{\tr\rho^2}\\
           +&\frac12\sum_{\ba}\frac{\left|(\tr\si^2)\,p_\si(\ba)-(\tr\rho^2)\,p_\si(\ba)\right|}{\tr\rho^2}\,.
    \end{align*}
    Using Eq.~\eqref{eq:tv0} to bound the first term of the previous inequality, and using the fact that $|\tr\rho^2-\tr\si^2|\le2\|\rho-\si\|_{\rm tr}$, Eq.~\eqref{eq:tv} follows.
\end{proof}

\begin{lemma}
\label{lemma:cont-CDF}
    Let $\rho,\rho'$ be quantum states with CDFs $F_\rho,F_{\rho'}$. Then, for all $\epsilon$
        \[
        F_{\rho'}(\epsilon)\le F_\rho(2\epsilon)+\frac{4\|\rho-\rho'\|_{\rm tr}}{\tr{\rho'}^2}
        \]
\end{lemma}

\begin{proof}
We consider the following subsets of Pauli strings:
\[
\begin{aligned}
\cA&=\{\ba:\langle P_{\ba}\rangle_{\rho'}^2\le \epsilon\}\\
\cB&=\{\ba:\langle P_{\ba}\rangle_\rho^2\le 2\epsilon\}.
\end{aligned}
\]
We further write $F_{\rho'}(\epsilon)=\sum_{\ba\in\cA}p_{\rho'}(\ba)$, and split the sum into two parts:
\begin{equation}\label{eq:2sums}
F_{\rho'}(\epsilon)=\sum_{\ba\in\cA\cap\cB}p_{\rho'}(\ba)+\sum_{\ba\in\cA\cap\bar\cB}p_{\rho'}(\ba),
\end{equation}
where $\bar\cB$ is the complement of $\cB$. Next, we upper-bound both contributions separately. For one,
\begin{multline}\label{eq:sumB}
\sum_{\ba\in\cA\cap\cB}p_{\rho'}(\ba)\le\sum_{\ba\in\cB}p_{\rho'}(\ba)\le F_\rho(2\epsilon)+{\rm TV}(p_\rho,p_{\rho'}),
\end{multline}
by the definition of the TV-distance.
Secondly,
\begin{align}\label{eq:sumBbar}
\sum_{\ba\in\cA\cap\bar\cB}p_{\rho'}(\ba)&\le\sum_{\ba\in\cA\cap\bar\cB}\frac{\langle P_\ba\rangle_\rho^2-\langle P_\ba\rangle_{\rho'}^2}{2^n\tr{\rho'}^2} \nonumber \\
& = \sum_{\ba\in\cA\cap\bar\cB}\frac{(\tr\rho^2)p_\rho(\ba)-(\tr{\rho'}^2)p_{\rho'}(\ba)}{\tr{\rho'}^2},
\end{align}
since $\langle P_\ba\rangle_\rho^2\ge2\langle P_\ba\rangle_{\rho'}^2$ whenever $\ba\in\cA\cap\bar\cB$. The statement follows by using Eqs.~\eqref{eq:tv} and~\eqref{eq:tv0} to upper-bound Eqs.~\eqref{eq:sumB} and~\eqref{eq:sumBbar}, respectively, and substituting into Eq.~\eqref{eq:2sums}.
\end{proof}

Let us now discuss the CDF of states in CW. Since the CDF is invariant under Cliffords, it suffices to consider the states $\ket{W_n}$ of Definition~\ref{def:CW}. A simple calculation yields the following Pauli expectation values:
\[
\bra{W_n}P_\ba\ket{W_n}=\begin{cases}
    1-\dfrac{2z_\ba}{n}\,,&(x_\ba,y_\ba)=(0,0)\\[2mm]
    \dfrac{2}{n}\,,&(x_\ba,y_\ba)\in\{(2,0),(0,2)\}\\[3mm]
    0\,,&{\rm otherwise}
\end{cases}
\]
where $x_\ba=|\{i:P_{a_i}=X\}|$ counts the number of tensor factors $X$ in $P_\ba=P_{a_1}\otimes\cdots\otimes P_{a_n}$, and similarly for $y_\ba$ and $z_\ba$. In particular,
\begin{equation}\label{eq:CDF-W}
    F_\rho(\ep)=0,
\end{equation}
for $\ep<4/n^2$ and $\rho\in{\rm CW}(n)$. Furthermore, for some $\rho\in{\rm CW}(n)$ and for arbitrary $\rho'$ with $\|\rho-\rho'\|_{\rm tr}\le1/10$, Lemma~\ref{lemma:cont-CDF} and Eq.~\eqref{eq:CDF-W} imply that for $\ep<2/n^2$
\begin{equation}\label{eq:CDF-W-approx}
    F_{\rho'}(\ep)\le\frac{4\|\rho-\rho'\|_{\rm tr}}{1-2\|\rho-\rho'\|_{\rm tr}}\le5\|\rho-\rho'\|_{\rm tr}\, .
\end{equation}

Lastly, an upper bound on the number of Pauli strings $P_\ba$ with $|\langle P_\ba\rangle_\rho|>3/4$ for $\rho \in {\rm CW}(n)$ will be useful in the proof of Proposition $\ref{prop:high-t}$. Since $|1-2z/n|>3/4$ if and only if $z<n/8$ or $z>7n/8$, one has
\begin{equation}\label{eq:counting-Ps}
    |\{\ba:|\langle P_\ba\rangle_\rho|>3/4\}|\leq 2\!\sum_{k=0}^{\lfloor n/8\rfloor}\binom{n}{k}\le 2^{nH(1/8)+1}<2^{3n/4},
\end{equation}
where $H(x)=-x\log x-(1-x)\log(1-x)$, by a well-known bound of the sum of binomial coefficients in terms of the binary entropy~\cite{galvin2014tutoriallecturesentropycounting}.

\section{Robustness of $\rm{rDIPE}$}\label{app:robustness}
In this section we prove the robustness result in the main text for states in ${\rm CW}$. To this end, consider two sequences $\{\rho'\}$, $\{\sigma'\}$ that are within trace distance $\tau$ of $\{\rho\}$, $\{\sigma\}\in{\rm CW}$. We have
\begin{align}\label{eq:Delta-bound}
    \Delta \equiv{\rm TV}(q'_{\rm mix},p'_{\rm mix})&\le\frac{{\rm TV}(q_{\rho_n'},p_{\rho_n'})}2+\frac{{\rm TV}(q_{\si_n'},p_{\si_n'})}2 \nonumber\\[2mm]
    &\hspace{-15mm}\le\frac{3\|\rho_n-\rho_n'\|_{\rm tr}}2+\frac{3\|\si_n-\si_n'\|_{\rm tr}}2 \leq 3\tau \,.
\end{align}
To derive the previous expression, we have used the triangle inequality ${\rm TV}(q_{\rho'},p_{\rho'})\le{\rm TV}(q_{\rho'},q_{\rho})+{\rm TV}(p_{\rho},p_{\rho'})$ since for real pure states $p_\rho=q_\rho$, and similarly for the term relative to $\sigma$. Finally, we used Eq.~\eqref{eq:tv} above and the fact that ${\rm TV}(q_{\rho'},q_{\rho})\le\|\rho-\rho'\|_{\rm tr}$ by the data processing inequality~\footnote{Indeed, ${\rm TV}(q_{\rho'},q_{\rho})\le\|\rho\otimes\rho-\rho'\otimes\rho'\|_{\rm tr}/2$ since the Bell distribution is the (diagonal) output of a quantum channel acting on two copies of a state, and $\|\rho\otimes\rho-\rho'\otimes\rho'\|_{\rm tr}\le2\|\rho-\rho'\|_{\rm tr}$}.

Now, for all $n,\tau>0$, consider $\ep_1,\ep_2$ given by $\ep_1=\tau/8$ and $\ep_2=\min\{(\tau/8)^2,1/n^2\}$, such that $4\ep_1+4\sqrt{\ep_2}\leq\tau$. From Eq.~\eqref{eq:CDF-W-approx} we have that the CDF $F(\epsilon)$ is small whenever $\epsilon$ is inverse-polynomially small. In particular, $F_{\rho'_n}(\ep_2)\le5\|\rho_n-\rho'_n\|_{\rm tr}\le5\tau$, since $\ep_2\leq 1/n^2 <2/n^2$. Similarly, $F_{\si'_n}(\ep_2)\le5\tau$ as well. Together with Eq.~\eqref{eq:Delta-bound}, this implies 
\begin{equation*}
4\ep_1+4\sqrt{\ep_2}+F_{\rho'_n}(\ep_2)+F_{\si'_n}(\ep_2)\\+6\Delta\le 29\tau \,.
\end{equation*}
Substituting into Eq.~\eqref{eq:performance} gives
\begin{multline*}\label{eq:sets-relation_robustness}
    \mathbb{P}\Big\{\left|f(\rho'_n,\si'_n,N_1,N_2)-{\rm c}(\rho'_n,\si'_n)\right|>29\tau\Big\}\\    <4\exp\left(-2\epsilon_1^2N_1\right)+4N_1\exp\left(-\epsilon_2^2N_2/2\right)\,.
\end{multline*}
The previous expression shows that there exist $N_1,N_2={\rm poly}(n,1/\tau,\de)$ such that Eq.~\eqref{eq:efficient-robust} holds when $\{\rho\},\{\si\}\in{\rm CW}$, for $k=29$.

\bibliography{refs}

\begin{thebibliography}{39}%
\makeatletter
\providecommand \@ifxundefined [1]{%
 \@ifx{#1\undefined}
}%
\providecommand \@ifnum [1]{%
 \ifnum #1\expandafter \@firstoftwo
 \else \expandafter \@secondoftwo
 \fi
}%
\providecommand \@ifx [1]{%
 \ifx #1\expandafter \@firstoftwo
 \else \expandafter \@secondoftwo
 \fi
}%
\providecommand \natexlab [1]{#1}%
\providecommand \enquote  [1]{``#1''}%
\providecommand \bibnamefont  [1]{#1}%
\providecommand \bibfnamefont [1]{#1}%
\providecommand \citenamefont [1]{#1}%
\providecommand \href@noop [0]{\@secondoftwo}%
\providecommand \href [0]{\begingroup \@sanitize@url \@href}%
\providecommand \@href[1]{\@@startlink{#1}\@@href}%
\providecommand \@@href[1]{\endgroup#1\@@endlink}%
\providecommand \@sanitize@url [0]{\catcode `\\12\catcode `\$12\catcode `\&12\catcode `\#12\catcode `\^12\catcode `\_12\catcode `\%12\relax}%
\providecommand \@@startlink[1]{}%
\providecommand \@@endlink[0]{}%
\providecommand \url  [0]{\begingroup\@sanitize@url \@url }%
\providecommand \@url [1]{\endgroup\@href {#1}{\urlprefix }}%
\providecommand \urlprefix  [0]{URL }%
\providecommand \Eprint [0]{\href }%
\providecommand \doibase [0]{https://doi.org/}%
\providecommand \selectlanguage [0]{\@gobble}%
\providecommand \bibinfo  [0]{\@secondoftwo}%
\providecommand \bibfield  [0]{\@secondoftwo}%
\providecommand \translation [1]{[#1]}%
\providecommand \BibitemOpen [0]{}%
\providecommand \bibitemStop [0]{}%
\providecommand \bibitemNoStop [0]{.\EOS\space}%
\providecommand \EOS [0]{\spacefactor3000\relax}%
\providecommand \BibitemShut  [1]{\csname bibitem#1\endcsname}%
\let\auto@bib@innerbib\@empty
\bibitem [{\citenamefont {Eisert}\ \emph {et~al.}(2020)\citenamefont {Eisert}, \citenamefont {Hangleiter}, \citenamefont {Walk}, \citenamefont {Roth}, \citenamefont {Markham}, \citenamefont {Parekh}, \citenamefont {Chabaud},\ and\ \citenamefont {Kashefi}}]{EHW20}%
  \BibitemOpen
  \bibfield  {author} {\bibinfo {author} {\bibfnamefont {J.}~\bibnamefont {Eisert}}, \bibinfo {author} {\bibfnamefont {D.}~\bibnamefont {Hangleiter}}, \bibinfo {author} {\bibfnamefont {N.}~\bibnamefont {Walk}}, \bibinfo {author} {\bibfnamefont {I.}~\bibnamefont {Roth}}, \bibinfo {author} {\bibfnamefont {D.}~\bibnamefont {Markham}}, \bibinfo {author} {\bibfnamefont {R.}~\bibnamefont {Parekh}}, \bibinfo {author} {\bibfnamefont {U.}~\bibnamefont {Chabaud}},\ and\ \bibinfo {author} {\bibfnamefont {E.}~\bibnamefont {Kashefi}},\ }\bibfield  {title} {\bibinfo {title} {Quantum certification and benchmarking},\ }\href {https://doi.org/10.1038/s42254-020-0186-4} {\bibfield  {journal} {\bibinfo  {journal} {Nature Reviews Physics}\ }\textbf {\bibinfo {volume} {2}},\ \bibinfo {pages} {382} (\bibinfo {year} {2020})}\BibitemShut {NoStop}%
\bibitem [{\citenamefont {O'Donnell}\ and\ \citenamefont {Wright}(2016)}]{OW16}%
  \BibitemOpen
  \bibfield  {author} {\bibinfo {author} {\bibfnamefont {R.}~\bibnamefont {O'Donnell}}\ and\ \bibinfo {author} {\bibfnamefont {J.}~\bibnamefont {Wright}},\ }\bibfield  {title} {\bibinfo {title} {Efficient quantum tomography},\ }in\ \href {https://doi.org/10.1145/2897518.2897544} {\emph {\bibinfo {booktitle} {Proceedings of the Forty-Eighth Annual {{ACM}} Symposium on {{Theory}} of {{Computing}}}}},\ \bibinfo {series and number} {{{STOC}} '16}\ (\bibinfo  {publisher} {Association for Computing Machinery},\ \bibinfo {address} {New York, NY, USA},\ \bibinfo {year} {2016})\ pp.\ \bibinfo {pages} {899--912}\BibitemShut {NoStop}%
\bibitem [{\citenamefont {{da Silva}}\ \emph {et~al.}(2011)\citenamefont {{da Silva}}, \citenamefont {{Landon-Cardinal}},\ and\ \citenamefont {Poulin}}]{DL11}%
  \BibitemOpen
  \bibfield  {author} {\bibinfo {author} {\bibfnamefont {M.~P.}\ \bibnamefont {{da Silva}}}, \bibinfo {author} {\bibfnamefont {O.}~\bibnamefont {{Landon-Cardinal}}},\ and\ \bibinfo {author} {\bibfnamefont {D.}~\bibnamefont {Poulin}},\ }\bibfield  {title} {\bibinfo {title} {Practical characterization of quantum devices without tomography},\ }\href {https://doi.org/10.1103/PhysRevLett.107.210404} {\bibfield  {journal} {\bibinfo  {journal} {Physical Review Letters}\ }\textbf {\bibinfo {volume} {107}},\ \bibinfo {pages} {210404} (\bibinfo {year} {2011})}\BibitemShut {NoStop}%
\bibitem [{\citenamefont {Flammia}\ and\ \citenamefont {Liu}(2011)}]{FL11}%
  \BibitemOpen
  \bibfield  {author} {\bibinfo {author} {\bibfnamefont {S.~T.}\ \bibnamefont {Flammia}}\ and\ \bibinfo {author} {\bibfnamefont {Y.-K.}\ \bibnamefont {Liu}},\ }\bibfield  {title} {\bibinfo {title} {Direct {{Fidelity Estimation}} from {{Few Pauli Measurements}}},\ }\href {https://doi.org/10.1103/PhysRevLett.106.230501} {\bibfield  {journal} {\bibinfo  {journal} {Physical Review Letters}\ }\textbf {\bibinfo {volume} {106}},\ \bibinfo {pages} {230501} (\bibinfo {year} {2011})}\BibitemShut {NoStop}%
\bibitem [{\citenamefont {Schuch}\ \emph {et~al.}(2008)\citenamefont {Schuch}, \citenamefont {Wolf}, \citenamefont {Verstraete},\ and\ \citenamefont {Cirac}}]{SWVC08}%
  \BibitemOpen
  \bibfield  {author} {\bibinfo {author} {\bibfnamefont {N.}~\bibnamefont {Schuch}}, \bibinfo {author} {\bibfnamefont {M.~M.}\ \bibnamefont {Wolf}}, \bibinfo {author} {\bibfnamefont {F.}~\bibnamefont {Verstraete}},\ and\ \bibinfo {author} {\bibfnamefont {J.~I.}\ \bibnamefont {Cirac}},\ }\bibfield  {title} {\bibinfo {title} {Entropy scaling and simulability by matrix product states},\ }\href {https://doi.org/10.1103/PhysRevLett.100.030504} {\bibfield  {journal} {\bibinfo  {journal} {Phys. Rev. Lett.}\ }\textbf {\bibinfo {volume} {100}},\ \bibinfo {pages} {030504} (\bibinfo {year} {2008})}\BibitemShut {NoStop}%
\bibitem [{\citenamefont {Chia}\ \emph {et~al.}(2024)\citenamefont {Chia}, \citenamefont {Lai},\ and\ \citenamefont {Lin}}]{CLL24}%
  \BibitemOpen
  \bibfield  {author} {\bibinfo {author} {\bibfnamefont {N.-H.}\ \bibnamefont {Chia}}, \bibinfo {author} {\bibfnamefont {C.-Y.}\ \bibnamefont {Lai}},\ and\ \bibinfo {author} {\bibfnamefont {H.-H.}\ \bibnamefont {Lin}},\ }\bibfield  {title} {\bibinfo {title} {Efficient learning of {$t$}-doped stabilizer states with single-copy measurements},\ }\href {https://doi.org/10.22331/q-2024-02-12-1250} {\bibfield  {journal} {\bibinfo  {journal} {{Quantum}}\ }\textbf {\bibinfo {volume} {8}},\ \bibinfo {pages} {1250} (\bibinfo {year} {2024})}\BibitemShut {NoStop}%
\bibitem [{\citenamefont {Leone}\ \emph {et~al.}(2024)\citenamefont {Leone}, \citenamefont {Oliviero},\ and\ \citenamefont {Hamma}}]{Leone2024learningtdoped}%
  \BibitemOpen
  \bibfield  {author} {\bibinfo {author} {\bibfnamefont {L.}~\bibnamefont {Leone}}, \bibinfo {author} {\bibfnamefont {S.~F.~E.}\ \bibnamefont {Oliviero}},\ and\ \bibinfo {author} {\bibfnamefont {A.}~\bibnamefont {Hamma}},\ }\bibfield  {title} {\bibinfo {title} {Learning t-doped stabilizer states},\ }\href {https://doi.org/10.22331/q-2024-05-27-1361} {\bibfield  {journal} {\bibinfo  {journal} {{Quantum}}\ }\textbf {\bibinfo {volume} {8}},\ \bibinfo {pages} {1361} (\bibinfo {year} {2024})}\BibitemShut {NoStop}%
\bibitem [{\citenamefont {Mahadev}(2018)}]{Ma18}%
  \BibitemOpen
  \bibfield  {author} {\bibinfo {author} {\bibfnamefont {U.}~\bibnamefont {Mahadev}},\ }\bibfield  {title} {\bibinfo {title} {Classical verification of quantum computations},\ }in\ \href {https://doi.org/10.1109/FOCS.2018.00033} {\emph {\bibinfo {booktitle} {{2018 IEEE 59th} {A}nnual {S}ymposium on {F}oundations of {C}omputer {S}cience ({FOCS})}}}\ (\bibinfo {year} {2018})\ pp.\ \bibinfo {pages} {259--267}\BibitemShut {NoStop}%
\bibitem [{\citenamefont {Brakerski}\ \emph {et~al.}(2020)\citenamefont {Brakerski}, \citenamefont {Koppula}, \citenamefont {Vazirani},\ and\ \citenamefont {Vidick}}]{BKVV20}%
  \BibitemOpen
  \bibfield  {author} {\bibinfo {author} {\bibfnamefont {Z.}~\bibnamefont {Brakerski}}, \bibinfo {author} {\bibfnamefont {V.}~\bibnamefont {Koppula}}, \bibinfo {author} {\bibfnamefont {U.}~\bibnamefont {Vazirani}},\ and\ \bibinfo {author} {\bibfnamefont {T.}~\bibnamefont {Vidick}},\ }\bibfield  {title} {\bibinfo {title} {Simpler proofs of quantumness},\ }in\ \href {https://doi.org/10.4230/LIPIcs.TQC.2020.8} {\emph {\bibinfo {booktitle} {15th Conference on the Theory of Quantum Computation, Communication and Cryptography (TQC 2020)}}},\ \bibinfo {series} {Leibniz International Proceedings in Informatics (LIPIcs)}, Vol.\ \bibinfo {volume} {158}\ (\bibinfo {address} {Dagstuhl, Germany},\ \bibinfo {year} {2020})\ pp.\ \bibinfo {pages} {8:1--8:14}\BibitemShut {NoStop}%
\bibitem [{\citenamefont {Brakerski}\ \emph {et~al.}(2021)\citenamefont {Brakerski}, \citenamefont {Christiano}, \citenamefont {Mahadev}, \citenamefont {Vazirani},\ and\ \citenamefont {Vidick}}]{BCMVV21}%
  \BibitemOpen
  \bibfield  {author} {\bibinfo {author} {\bibfnamefont {Z.}~\bibnamefont {Brakerski}}, \bibinfo {author} {\bibfnamefont {P.}~\bibnamefont {Christiano}}, \bibinfo {author} {\bibfnamefont {U.}~\bibnamefont {Mahadev}}, \bibinfo {author} {\bibfnamefont {U.}~\bibnamefont {Vazirani}},\ and\ \bibinfo {author} {\bibfnamefont {T.}~\bibnamefont {Vidick}},\ }\bibfield  {title} {\bibinfo {title} {A cryptographic test of quantumness and certifiable randomness from a single quantum device},\ }\href {https://doi.org/10.1145/3441309} {\bibfield  {journal} {\bibinfo  {journal} {Journal of the ACM}\ }\textbf {\bibinfo {volume} {68}},\ \bibinfo {pages} {1} (\bibinfo {year} {2021})}\BibitemShut {NoStop}%
\bibitem [{\citenamefont {Elben}\ \emph {et~al.}(2020)\citenamefont {Elben}, \citenamefont {Vermersch}, \citenamefont {{van Bijnen}}, \citenamefont {Kokail}, \citenamefont {Brydges}, \citenamefont {Maier}, \citenamefont {Joshi}, \citenamefont {Blatt}, \citenamefont {Roos},\ and\ \citenamefont {Zoller}}]{El20}%
  \BibitemOpen
  \bibfield  {author} {\bibinfo {author} {\bibfnamefont {A.}~\bibnamefont {Elben}}, \bibinfo {author} {\bibfnamefont {B.}~\bibnamefont {Vermersch}}, \bibinfo {author} {\bibfnamefont {R.}~\bibnamefont {{van Bijnen}}}, \bibinfo {author} {\bibfnamefont {C.}~\bibnamefont {Kokail}}, \bibinfo {author} {\bibfnamefont {T.}~\bibnamefont {Brydges}}, \bibinfo {author} {\bibfnamefont {C.}~\bibnamefont {Maier}}, \bibinfo {author} {\bibfnamefont {M.}~\bibnamefont {Joshi}}, \bibinfo {author} {\bibfnamefont {R.}~\bibnamefont {Blatt}}, \bibinfo {author} {\bibfnamefont {C.~F.}\ \bibnamefont {Roos}},\ and\ \bibinfo {author} {\bibfnamefont {P.}~\bibnamefont {Zoller}},\ }\bibfield  {title} {\bibinfo {title} {Cross-{{Platform Verification}} of {{Intermediate Scale Quantum Devices}}},\ }\href {https://doi.org/10.1103/PhysRevLett.124.010504} {\bibfield  {journal} {\bibinfo  {journal} {Physical Review Letters}\ }\textbf {\bibinfo {volume} {124}},\ \bibinfo {pages} {010504} (\bibinfo {year} {2020})}\BibitemShut {NoStop}%
\bibitem [{\citenamefont {Carrasco}\ \emph {et~al.}(2021)\citenamefont {Carrasco}, \citenamefont {Elben}, \citenamefont {Kokail}, \citenamefont {Kraus},\ and\ \citenamefont {Zoller}}]{CEKKZ21}%
  \BibitemOpen
  \bibfield  {author} {\bibinfo {author} {\bibfnamefont {J.}~\bibnamefont {Carrasco}}, \bibinfo {author} {\bibfnamefont {A.}~\bibnamefont {Elben}}, \bibinfo {author} {\bibfnamefont {C.}~\bibnamefont {Kokail}}, \bibinfo {author} {\bibfnamefont {B.}~\bibnamefont {Kraus}},\ and\ \bibinfo {author} {\bibfnamefont {P.}~\bibnamefont {Zoller}},\ }\bibfield  {title} {\bibinfo {title} {Theoretical and {{Experimental Perspectives}} of {{Quantum Verification}}},\ }\href {https://doi.org/10.1103/PRXQuantum.2.010102} {\bibfield  {journal} {\bibinfo  {journal} {PRX Quantum}\ }\textbf {\bibinfo {volume} {2}},\ \bibinfo {pages} {010102} (\bibinfo {year} {2021})}\BibitemShut {NoStop}%
\bibitem [{\citenamefont {Anshu}\ \emph {et~al.}(2022)\citenamefont {Anshu}, \citenamefont {Landau},\ and\ \citenamefont {Liu}}]{ALL22}%
  \BibitemOpen
  \bibfield  {author} {\bibinfo {author} {\bibfnamefont {A.}~\bibnamefont {Anshu}}, \bibinfo {author} {\bibfnamefont {Z.}~\bibnamefont {Landau}},\ and\ \bibinfo {author} {\bibfnamefont {Y.}~\bibnamefont {Liu}},\ }\bibfield  {title} {\bibinfo {title} {Distributed quantum inner product estimation},\ }in\ \href {https://doi.org/10.1145/3519935.3519974} {\emph {\bibinfo {booktitle} {Proceedings of the 54th {{Annual ACM SIGACT Symposium}} on {{Theory}} of {{Computing}}}}}\ (\bibinfo {year} {2022})\ pp.\ \bibinfo {pages} {44--51}\BibitemShut {NoStop}%
\bibitem [{\citenamefont {Kn\"orzer}\ \emph {et~al.}(2023)\citenamefont {Kn\"orzer}, \citenamefont {Malz},\ and\ \citenamefont {Cirac}}]{Kn23}%
  \BibitemOpen
  \bibfield  {author} {\bibinfo {author} {\bibfnamefont {J.}~\bibnamefont {Kn\"orzer}}, \bibinfo {author} {\bibfnamefont {D.}~\bibnamefont {Malz}},\ and\ \bibinfo {author} {\bibfnamefont {J.~I.}\ \bibnamefont {Cirac}},\ }\bibfield  {title} {\bibinfo {title} {Cross-platform verification in quantum networks},\ }\href {https://doi.org/10.1103/PhysRevA.107.062424} {\bibfield  {journal} {\bibinfo  {journal} {Phys. Rev. A}\ }\textbf {\bibinfo {volume} {107}},\ \bibinfo {pages} {062424} (\bibinfo {year} {2023})}\BibitemShut {NoStop}%
\bibitem [{\citenamefont {Gong}\ \emph {et~al.}(2024)\citenamefont {Gong}, \citenamefont {Haferkamp}, \citenamefont {Ye},\ and\ \citenamefont {Zhang}}]{Go24}%
  \BibitemOpen
  \bibfield  {author} {\bibinfo {author} {\bibfnamefont {W.}~\bibnamefont {Gong}}, \bibinfo {author} {\bibfnamefont {J.}~\bibnamefont {Haferkamp}}, \bibinfo {author} {\bibfnamefont {Q.}~\bibnamefont {Ye}},\ and\ \bibinfo {author} {\bibfnamefont {Z.}~\bibnamefont {Zhang}},\ }\href@noop {} {\bibinfo {title} {On the sample complexity of purity and inner product estimation}} (\bibinfo {year} {2024}),\ \Eprint {https://arxiv.org/abs/2410.12712} {arXiv:2410.12712 [quant-ph]} \BibitemShut {NoStop}%
\bibitem [{\citenamefont {Arunachalam}\ and\ \citenamefont {Schatzki}(2024)}]{Ar24}%
  \BibitemOpen
  \bibfield  {author} {\bibinfo {author} {\bibfnamefont {S.}~\bibnamefont {Arunachalam}}\ and\ \bibinfo {author} {\bibfnamefont {L.}~\bibnamefont {Schatzki}},\ }\href@noop {} {\bibinfo {title} {Distributed inner product estimation with limited quantum communication}} (\bibinfo {year} {2024}),\ \Eprint {https://arxiv.org/abs/2410.12684} {arXiv:2410.12684 [quant-ph]} \BibitemShut {NoStop}%
\bibitem [{\citenamefont {Qian}\ \emph {et~al.}(2024)\citenamefont {Qian}, \citenamefont {Du}, \citenamefont {He}, \citenamefont {Hsieh},\ and\ \citenamefont {Tao}}]{NN24}%
  \BibitemOpen
  \bibfield  {author} {\bibinfo {author} {\bibfnamefont {Y.}~\bibnamefont {Qian}}, \bibinfo {author} {\bibfnamefont {Y.}~\bibnamefont {Du}}, \bibinfo {author} {\bibfnamefont {Z.}~\bibnamefont {He}}, \bibinfo {author} {\bibfnamefont {M.-H.}\ \bibnamefont {Hsieh}},\ and\ \bibinfo {author} {\bibfnamefont {D.}~\bibnamefont {Tao}},\ }\bibfield  {title} {\bibinfo {title} {Multimodal deep representation learning for quantum cross-platform verification},\ }\href {https://doi.org/10.1103/PhysRevLett.133.130601} {\bibfield  {journal} {\bibinfo  {journal} {Phys. Rev. Lett.}\ }\textbf {\bibinfo {volume} {133}},\ \bibinfo {pages} {130601} (\bibinfo {year} {2024})}\BibitemShut {NoStop}%
\bibitem [{\citenamefont {Hinsche}\ \emph {et~al.}(2024)\citenamefont {Hinsche}, \citenamefont {Ioannou}, \citenamefont {Jerbi}, \citenamefont {Leone}, \citenamefont {Eisert},\ and\ \citenamefont {Carrasco}}]{HIJLEC24}%
  \BibitemOpen
  \bibfield  {author} {\bibinfo {author} {\bibfnamefont {M.}~\bibnamefont {Hinsche}}, \bibinfo {author} {\bibfnamefont {M.}~\bibnamefont {Ioannou}}, \bibinfo {author} {\bibfnamefont {S.}~\bibnamefont {Jerbi}}, \bibinfo {author} {\bibfnamefont {L.}~\bibnamefont {Leone}}, \bibinfo {author} {\bibfnamefont {J.}~\bibnamefont {Eisert}},\ and\ \bibinfo {author} {\bibfnamefont {J.}~\bibnamefont {Carrasco}},\ }\href@noop {} {\bibinfo {title} {Efficient distributed inner product estimation via {P}auli sampling}} (\bibinfo {year} {2024}),\ \Eprint {https://arxiv.org/abs/2405.06544} {arXiv:2405.06544 [quant-ph]} \BibitemShut {NoStop}%
\bibitem [{\citenamefont {Montanaro}(2017)}]{Mo17}%
  \BibitemOpen
  \bibfield  {author} {\bibinfo {author} {\bibfnamefont {A.}~\bibnamefont {Montanaro}},\ }\href@noop {} {\bibinfo {title} {Learning stabilizer states by bell sampling}} (\bibinfo {year} {2017}),\ \Eprint {https://arxiv.org/abs/1707.04012} {arXiv:1707.04012 [quant-ph]} \BibitemShut {NoStop}%
\bibitem [{Note1()}]{Note1}%
  \BibitemOpen
  \bibinfo {note} {Here and in what follows, ${\protect \mathbb P}$ will denote the probability with respect to the internal randomness of rDIPE for fixed inputs.}\BibitemShut {Stop}%
\bibitem [{\citenamefont {Hashagen}\ \emph {et~al.}(2018)\citenamefont {Hashagen}, \citenamefont {Flammia}, \citenamefont {Gross},\ and\ \citenamefont {Wallman}}]{HFGW18}%
  \BibitemOpen
  \bibfield  {author} {\bibinfo {author} {\bibfnamefont {A.~K.}\ \bibnamefont {Hashagen}}, \bibinfo {author} {\bibfnamefont {S.~T.}\ \bibnamefont {Flammia}}, \bibinfo {author} {\bibfnamefont {D.}~\bibnamefont {Gross}},\ and\ \bibinfo {author} {\bibfnamefont {J.~J.}\ \bibnamefont {Wallman}},\ }\bibfield  {title} {\bibinfo {title} {Real {R}andomized {B}enchmarking},\ }\href {https://doi.org/10.22331/q-2018-08-22-85} {\bibfield  {journal} {\bibinfo  {journal} {{Quantum}}\ }\textbf {\bibinfo {volume} {2}},\ \bibinfo {pages} {85} (\bibinfo {year} {2018})}\BibitemShut {NoStop}%
\bibitem [{\citenamefont {Ruiz}\ \emph {et~al.}(2024)\citenamefont {Ruiz}, \citenamefont {Sopena}, \citenamefont {Pozsgay},\ and\ \citenamefont {López}}]{RSPL24}%
  \BibitemOpen
  \bibfield  {author} {\bibinfo {author} {\bibfnamefont {R.}~\bibnamefont {Ruiz}}, \bibinfo {author} {\bibfnamefont {A.}~\bibnamefont {Sopena}}, \bibinfo {author} {\bibfnamefont {B.}~\bibnamefont {Pozsgay}},\ and\ \bibinfo {author} {\bibfnamefont {E.}~\bibnamefont {López}},\ }\href@noop {} {\bibinfo {title} {Efficient eigenstate preparation in an integrable model with {H}ilbert space fragmentation}} (\bibinfo {year} {2024}),\ \Eprint {https://arxiv.org/abs/2411.15132} {arXiv:2411.15132 [quant-ph]} \BibitemShut {NoStop}%
\bibitem [{\citenamefont {Hauschild}\ \emph {et~al.}(2024)\citenamefont {Hauschild}, \citenamefont {Unfried}, \citenamefont {Anand}, \citenamefont {Andrews}, \citenamefont {Bintz}, \citenamefont {Borla}, \citenamefont {Divic}, \citenamefont {Drescher}, \citenamefont {Geiger}, \citenamefont {Hefel}, \citenamefont {Hémery}, \citenamefont {Kadow}, \citenamefont {Kemp}, \citenamefont {Kirchner}, \citenamefont {Liu}, \citenamefont {Möller}, \citenamefont {Parker}, \citenamefont {Rader}, \citenamefont {Romen}, \citenamefont {Scalet}, \citenamefont {Schoonderwoerd}, \citenamefont {Schulz}, \citenamefont {Soejima}, \citenamefont {Thoma}, \citenamefont {Wu}, \citenamefont {Zechmann}, \citenamefont {Zweng}, \citenamefont {Mong}, \citenamefont {Zaletel},\ and\ \citenamefont {Pollmann}}]{TenPy}%
  \BibitemOpen
  \bibfield  {author} {\bibinfo {author} {\bibfnamefont {J.}~\bibnamefont {Hauschild}}, \bibinfo {author} {\bibfnamefont {J.}~\bibnamefont {Unfried}}, \bibinfo {author} {\bibfnamefont {S.}~\bibnamefont {Anand}}, \bibinfo {author} {\bibfnamefont {B.}~\bibnamefont {Andrews}}, \bibinfo {author} {\bibfnamefont {M.}~\bibnamefont {Bintz}}, \bibinfo {author} {\bibfnamefont {U.}~\bibnamefont {Borla}}, \bibinfo {author} {\bibfnamefont {S.}~\bibnamefont {Divic}}, \bibinfo {author} {\bibfnamefont {M.}~\bibnamefont {Drescher}}, \bibinfo {author} {\bibfnamefont {J.}~\bibnamefont {Geiger}}, \bibinfo {author} {\bibfnamefont {M.}~\bibnamefont {Hefel}}, \bibinfo {author} {\bibfnamefont {K.}~\bibnamefont {Hémery}}, \bibinfo {author} {\bibfnamefont {W.}~\bibnamefont {Kadow}}, \bibinfo {author} {\bibfnamefont {J.}~\bibnamefont {Kemp}}, \bibinfo {author} {\bibfnamefont {N.}~\bibnamefont {Kirchner}}, \bibinfo {author} {\bibfnamefont {V.~S.}\ \bibnamefont {Liu}}, \bibinfo {author} {\bibfnamefont {G.}~\bibnamefont {Möller}},
  \bibinfo {author} {\bibfnamefont {D.}~\bibnamefont {Parker}}, \bibinfo {author} {\bibfnamefont {M.}~\bibnamefont {Rader}}, \bibinfo {author} {\bibfnamefont {A.}~\bibnamefont {Romen}}, \bibinfo {author} {\bibfnamefont {S.}~\bibnamefont {Scalet}}, \bibinfo {author} {\bibfnamefont {L.}~\bibnamefont {Schoonderwoerd}}, \bibinfo {author} {\bibfnamefont {M.}~\bibnamefont {Schulz}}, \bibinfo {author} {\bibfnamefont {T.}~\bibnamefont {Soejima}}, \bibinfo {author} {\bibfnamefont {P.}~\bibnamefont {Thoma}}, \bibinfo {author} {\bibfnamefont {Y.}~\bibnamefont {Wu}}, \bibinfo {author} {\bibfnamefont {P.}~\bibnamefont {Zechmann}}, \bibinfo {author} {\bibfnamefont {L.}~\bibnamefont {Zweng}}, \bibinfo {author} {\bibfnamefont {R.~S.~K.}\ \bibnamefont {Mong}}, \bibinfo {author} {\bibfnamefont {M.~P.}\ \bibnamefont {Zaletel}},\ and\ \bibinfo {author} {\bibfnamefont {F.}~\bibnamefont {Pollmann}},\ }\bibfield  {title} {\bibinfo {title} {{Tensor network Python (TeNPy) version 1}},\ }\href
  {https://doi.org/10.21468/SciPostPhysCodeb.41} {\bibfield  {journal} {\bibinfo  {journal} {SciPost Phys. Codebases}\ ,\ \bibinfo {pages} {41}} (\bibinfo {year} {2024})}\BibitemShut {NoStop}%
\bibitem [{Note2()}]{Note2}%
  \BibitemOpen
  \bibinfo {note} {The proof is analogous for arbitrary non-Clifford gates acting on a constant number of qubits}\BibitemShut {NoStop}%
\bibitem [{\citenamefont {Aharonov}(2003)}]{aharonov2003simple}%
  \BibitemOpen
  \bibfield  {author} {\bibinfo {author} {\bibfnamefont {D.}~\bibnamefont {Aharonov}},\ }\href@noop {} {\bibinfo {title} {A simple proof that {T}offoli and {H}adamard are quantum universal}} (\bibinfo {year} {2003}),\ \Eprint {https://arxiv.org/abs/quant-ph/0301040} {arXiv:quant-ph/0301040 [quant-ph]} \BibitemShut {NoStop}%
\bibitem [{\citenamefont {Shi}(2002)}]{shi2002toffoli}%
  \BibitemOpen
  \bibfield  {author} {\bibinfo {author} {\bibfnamefont {Y.}~\bibnamefont {Shi}},\ }\href@noop {} {\bibinfo {title} {Both {T}offoli and {C}ontrolled-{NOT} need little help to do universal quantum computation}} (\bibinfo {year} {2002}),\ \Eprint {https://arxiv.org/abs/quant-ph/0205115} {arXiv:quant-ph/0205115 [quant-ph]} \BibitemShut {NoStop}%
\bibitem [{\citenamefont {Wallman}\ and\ \citenamefont {Emerson}(2016)}]{WE16}%
  \BibitemOpen
  \bibfield  {author} {\bibinfo {author} {\bibfnamefont {J.~J.}\ \bibnamefont {Wallman}}\ and\ \bibinfo {author} {\bibfnamefont {J.}~\bibnamefont {Emerson}},\ }\bibfield  {title} {\bibinfo {title} {Noise tailoring for scalable quantum computation via randomized compiling},\ }\bibfield  {journal} {\bibinfo  {journal} {Physical Review A}\ }\textbf {\bibinfo {volume} {94}},\ \href {https://doi.org/10.1103/physreva.94.052325} {10.1103/physreva.94.052325} (\bibinfo {year} {2016})\BibitemShut {NoStop}%
\bibitem [{\citenamefont {Dvoretzky}\ \emph {et~al.}(1956)\citenamefont {Dvoretzky}, \citenamefont {Kiefer},\ and\ \citenamefont {Wolfowitz}}]{dvoretzky1956asymptotic}%
  \BibitemOpen
  \bibfield  {author} {\bibinfo {author} {\bibfnamefont {A.}~\bibnamefont {Dvoretzky}}, \bibinfo {author} {\bibfnamefont {J.}~\bibnamefont {Kiefer}},\ and\ \bibinfo {author} {\bibfnamefont {J.}~\bibnamefont {Wolfowitz}},\ }\bibfield  {title} {\bibinfo {title} {Asymptotic minimax character of the sample distribution function and of the classical multinomial estimator},\ }\href@noop {} {\bibfield  {journal} {\bibinfo  {journal} {The Annals of Mathematical Statistics}\ }\textbf {\bibinfo {volume} {27}},\ \bibinfo {pages} {642} (\bibinfo {year} {1956})}\BibitemShut {NoStop}%
\bibitem [{\citenamefont {Massart}(1990)}]{DKW}%
  \BibitemOpen
  \bibfield  {author} {\bibinfo {author} {\bibfnamefont {P.}~\bibnamefont {Massart}},\ }\bibfield  {title} {\bibinfo {title} {{The Tight Constant in the Dvoretzky-Kiefer-Wolfowitz Inequality}},\ }\href {https://doi.org/10.1214/aop/1176990746} {\bibfield  {journal} {\bibinfo  {journal} {The Annals of Probability}\ }\textbf {\bibinfo {volume} {18}},\ \bibinfo {pages} {1269 } (\bibinfo {year} {1990})}\BibitemShut {NoStop}%
\bibitem [{\citenamefont {Huang}\ \emph {et~al.}(2022)\citenamefont {Huang}, \citenamefont {Broughton}, \citenamefont {Cotler}, \citenamefont {Chen}, \citenamefont {Li}, \citenamefont {Mohseni}, \citenamefont {Neven}, \citenamefont {Babbush}, \citenamefont {Kueng}, \citenamefont {Preskill},\ and\ \citenamefont {McClean}}]{Science-two-copies-theory}%
  \BibitemOpen
  \bibfield  {author} {\bibinfo {author} {\bibfnamefont {H.-Y.}\ \bibnamefont {Huang}}, \bibinfo {author} {\bibfnamefont {M.}~\bibnamefont {Broughton}}, \bibinfo {author} {\bibfnamefont {J.}~\bibnamefont {Cotler}}, \bibinfo {author} {\bibfnamefont {S.}~\bibnamefont {Chen}}, \bibinfo {author} {\bibfnamefont {J.}~\bibnamefont {Li}}, \bibinfo {author} {\bibfnamefont {M.}~\bibnamefont {Mohseni}}, \bibinfo {author} {\bibfnamefont {H.}~\bibnamefont {Neven}}, \bibinfo {author} {\bibfnamefont {R.}~\bibnamefont {Babbush}}, \bibinfo {author} {\bibfnamefont {R.}~\bibnamefont {Kueng}}, \bibinfo {author} {\bibfnamefont {J.}~\bibnamefont {Preskill}},\ and\ \bibinfo {author} {\bibfnamefont {J.~R.}\ \bibnamefont {McClean}},\ }\bibfield  {title} {\bibinfo {title} {Quantum advantage in learning from experiments},\ }\href {https://doi.org/10.1126/science.abn7293} {\bibfield  {journal} {\bibinfo  {journal} {Science}\ }\textbf {\bibinfo {volume} {376}},\ \bibinfo {pages} {1182} (\bibinfo {year} {2022})}\BibitemShut {NoStop}%
\bibitem [{\citenamefont {Schmid}\ \emph {et~al.}(2008)\citenamefont {Schmid}, \citenamefont {Kiesel}, \citenamefont {Wieczorek}, \citenamefont {Weinfurter}, \citenamefont {Mintert},\ and\ \citenamefont {Buchleitner}}]{two-copies-photons}%
  \BibitemOpen
  \bibfield  {author} {\bibinfo {author} {\bibfnamefont {C.}~\bibnamefont {Schmid}}, \bibinfo {author} {\bibfnamefont {N.}~\bibnamefont {Kiesel}}, \bibinfo {author} {\bibfnamefont {W.}~\bibnamefont {Wieczorek}}, \bibinfo {author} {\bibfnamefont {H.}~\bibnamefont {Weinfurter}}, \bibinfo {author} {\bibfnamefont {F.}~\bibnamefont {Mintert}},\ and\ \bibinfo {author} {\bibfnamefont {A.}~\bibnamefont {Buchleitner}},\ }\bibfield  {title} {\bibinfo {title} {Experimental direct observation of mixed state entanglement},\ }\href {https://doi.org/10.1103/PhysRevLett.101.260505} {\bibfield  {journal} {\bibinfo  {journal} {Phys. Rev. Lett.}\ }\textbf {\bibinfo {volume} {101}},\ \bibinfo {pages} {260505} (\bibinfo {year} {2008})}\BibitemShut {NoStop}%
\bibitem [{\citenamefont {Miller}\ \emph {et~al.}(2024)\citenamefont {Miller}, \citenamefont {Levi}, \citenamefont {Postler}, \citenamefont {Steiner}, \citenamefont {Bittel}, \citenamefont {White}, \citenamefont {Tang}, \citenamefont {Kuehnke}, \citenamefont {Mele}, \citenamefont {Khatri}, \citenamefont {Leone}, \citenamefont {Carrasco}, \citenamefont {Marciniak}, \citenamefont {Pogorelov}, \citenamefont {Guevara-Bertsch}, \citenamefont {Freund}, \citenamefont {Blatt}, \citenamefont {Schindler}, \citenamefont {Monz}, \citenamefont {Ringbauer},\ and\ \citenamefont {Eisert}}]{two-copies-trapped-ions}%
  \BibitemOpen
  \bibfield  {author} {\bibinfo {author} {\bibfnamefont {D.}~\bibnamefont {Miller}}, \bibinfo {author} {\bibfnamefont {K.}~\bibnamefont {Levi}}, \bibinfo {author} {\bibfnamefont {L.}~\bibnamefont {Postler}}, \bibinfo {author} {\bibfnamefont {A.}~\bibnamefont {Steiner}}, \bibinfo {author} {\bibfnamefont {L.}~\bibnamefont {Bittel}}, \bibinfo {author} {\bibfnamefont {G.~A.~L.}\ \bibnamefont {White}}, \bibinfo {author} {\bibfnamefont {Y.}~\bibnamefont {Tang}}, \bibinfo {author} {\bibfnamefont {E.~J.}\ \bibnamefont {Kuehnke}}, \bibinfo {author} {\bibfnamefont {A.~A.}\ \bibnamefont {Mele}}, \bibinfo {author} {\bibfnamefont {S.}~\bibnamefont {Khatri}}, \bibinfo {author} {\bibfnamefont {L.}~\bibnamefont {Leone}}, \bibinfo {author} {\bibfnamefont {J.}~\bibnamefont {Carrasco}}, \bibinfo {author} {\bibfnamefont {C.~D.}\ \bibnamefont {Marciniak}}, \bibinfo {author} {\bibfnamefont {I.}~\bibnamefont {Pogorelov}}, \bibinfo {author} {\bibfnamefont {M.}~\bibnamefont {Guevara-Bertsch}}, \bibinfo {author} {\bibfnamefont
  {R.}~\bibnamefont {Freund}}, \bibinfo {author} {\bibfnamefont {R.}~\bibnamefont {Blatt}}, \bibinfo {author} {\bibfnamefont {P.}~\bibnamefont {Schindler}}, \bibinfo {author} {\bibfnamefont {T.}~\bibnamefont {Monz}}, \bibinfo {author} {\bibfnamefont {M.}~\bibnamefont {Ringbauer}},\ and\ \bibinfo {author} {\bibfnamefont {J.}~\bibnamefont {Eisert}},\ }\href@noop {} {\bibinfo {title} {Experimental measurement and a physical interpretation of quantum shadow enumerators}} (\bibinfo {year} {2024}),\ \Eprint {https://arxiv.org/abs/2408.16914} {arXiv:2408.16914 [quant-ph]} \BibitemShut {NoStop}%
\bibitem [{\citenamefont {Islam}\ \emph {et~al.}(2015)\citenamefont {Islam}, \citenamefont {Ma}, \citenamefont {Preiss}, \citenamefont {Eric~Tai}, \citenamefont {Lukin}, \citenamefont {Rispoli},\ and\ \citenamefont {Greiner}}]{two-copies-ultra-cold-bosonic-atoms}%
  \BibitemOpen
  \bibfield  {author} {\bibinfo {author} {\bibfnamefont {R.}~\bibnamefont {Islam}}, \bibinfo {author} {\bibfnamefont {R.}~\bibnamefont {Ma}}, \bibinfo {author} {\bibfnamefont {P.~M.}\ \bibnamefont {Preiss}}, \bibinfo {author} {\bibfnamefont {M.}~\bibnamefont {Eric~Tai}}, \bibinfo {author} {\bibfnamefont {A.}~\bibnamefont {Lukin}}, \bibinfo {author} {\bibfnamefont {M.}~\bibnamefont {Rispoli}},\ and\ \bibinfo {author} {\bibfnamefont {M.}~\bibnamefont {Greiner}},\ }\bibfield  {title} {\bibinfo {title} {Measuring entanglement entropy in a quantum many-body system},\ }\href {https://doi.org/10.1038/nature15750} {\bibfield  {journal} {\bibinfo  {journal} {Nature}\ }\textbf {\bibinfo {volume} {528}},\ \bibinfo {pages} {77–83} (\bibinfo {year} {2015})}\BibitemShut {NoStop}%
\bibitem [{\citenamefont {Bluvstein}\ \emph {et~al.}(2022)\citenamefont {Bluvstein}, \citenamefont {Levine}, \citenamefont {Semeghini}, \citenamefont {Wang}, \citenamefont {Ebadi}, \citenamefont {Kalinowski}, \citenamefont {Keesling}, \citenamefont {Maskara}, \citenamefont {Pichler}, \citenamefont {Greiner}, \citenamefont {Vuletić},\ and\ \citenamefont {Lukin}}]{two-copies-lukin}%
  \BibitemOpen
  \bibfield  {author} {\bibinfo {author} {\bibfnamefont {D.}~\bibnamefont {Bluvstein}}, \bibinfo {author} {\bibfnamefont {H.}~\bibnamefont {Levine}}, \bibinfo {author} {\bibfnamefont {G.}~\bibnamefont {Semeghini}}, \bibinfo {author} {\bibfnamefont {T.~T.}\ \bibnamefont {Wang}}, \bibinfo {author} {\bibfnamefont {S.}~\bibnamefont {Ebadi}}, \bibinfo {author} {\bibfnamefont {M.}~\bibnamefont {Kalinowski}}, \bibinfo {author} {\bibfnamefont {A.}~\bibnamefont {Keesling}}, \bibinfo {author} {\bibfnamefont {N.}~\bibnamefont {Maskara}}, \bibinfo {author} {\bibfnamefont {H.}~\bibnamefont {Pichler}}, \bibinfo {author} {\bibfnamefont {M.}~\bibnamefont {Greiner}}, \bibinfo {author} {\bibfnamefont {V.}~\bibnamefont {Vuletić}},\ and\ \bibinfo {author} {\bibfnamefont {M.~D.}\ \bibnamefont {Lukin}},\ }\bibfield  {title} {\bibinfo {title} {A quantum processor based on coherent transport of entangled atom arrays},\ }\href {https://doi.org/10.1038/s41586-022-04592-6} {\bibfield  {journal} {\bibinfo  {journal} {Nature}\ }\textbf
  {\bibinfo {volume} {604}},\ \bibinfo {pages} {451–456} (\bibinfo {year} {2022})}\BibitemShut {NoStop}%
\bibitem [{\citenamefont {Bluvstein}\ \emph {et~al.}(2023)\citenamefont {Bluvstein}, \citenamefont {Evered}, \citenamefont {Geim}, \citenamefont {Li}, \citenamefont {Zhou}, \citenamefont {Manovitz}, \citenamefont {Ebadi}, \citenamefont {Cain}, \citenamefont {Kalinowski}, \citenamefont {Hangleiter}, \citenamefont {Bonilla~Ataides}, \citenamefont {Maskara}, \citenamefont {Cong}, \citenamefont {Gao}, \citenamefont {Sales~Rodriguez}, \citenamefont {Karolyshyn}, \citenamefont {Semeghini}, \citenamefont {Gullans}, \citenamefont {Greiner}, \citenamefont {Vuletić},\ and\ \citenamefont {Lukin}}]{two-logical-copies-lukin}%
  \BibitemOpen
  \bibfield  {author} {\bibinfo {author} {\bibfnamefont {D.}~\bibnamefont {Bluvstein}}, \bibinfo {author} {\bibfnamefont {S.~J.}\ \bibnamefont {Evered}}, \bibinfo {author} {\bibfnamefont {A.~A.}\ \bibnamefont {Geim}}, \bibinfo {author} {\bibfnamefont {S.~H.}\ \bibnamefont {Li}}, \bibinfo {author} {\bibfnamefont {H.}~\bibnamefont {Zhou}}, \bibinfo {author} {\bibfnamefont {T.}~\bibnamefont {Manovitz}}, \bibinfo {author} {\bibfnamefont {S.}~\bibnamefont {Ebadi}}, \bibinfo {author} {\bibfnamefont {M.}~\bibnamefont {Cain}}, \bibinfo {author} {\bibfnamefont {M.}~\bibnamefont {Kalinowski}}, \bibinfo {author} {\bibfnamefont {D.}~\bibnamefont {Hangleiter}}, \bibinfo {author} {\bibfnamefont {J.~P.}\ \bibnamefont {Bonilla~Ataides}}, \bibinfo {author} {\bibfnamefont {N.}~\bibnamefont {Maskara}}, \bibinfo {author} {\bibfnamefont {I.}~\bibnamefont {Cong}}, \bibinfo {author} {\bibfnamefont {X.}~\bibnamefont {Gao}}, \bibinfo {author} {\bibfnamefont {P.}~\bibnamefont {Sales~Rodriguez}}, \bibinfo {author} {\bibfnamefont
  {T.}~\bibnamefont {Karolyshyn}}, \bibinfo {author} {\bibfnamefont {G.}~\bibnamefont {Semeghini}}, \bibinfo {author} {\bibfnamefont {M.~J.}\ \bibnamefont {Gullans}}, \bibinfo {author} {\bibfnamefont {M.}~\bibnamefont {Greiner}}, \bibinfo {author} {\bibfnamefont {V.}~\bibnamefont {Vuletić}},\ and\ \bibinfo {author} {\bibfnamefont {M.~D.}\ \bibnamefont {Lukin}},\ }\bibfield  {title} {\bibinfo {title} {Logical quantum processor based on reconfigurable atom arrays},\ }\href {https://doi.org/10.1038/s41586-023-06927-3} {\bibfield  {journal} {\bibinfo  {journal} {Nature}\ }\textbf {\bibinfo {volume} {626}},\ \bibinfo {pages} {58–65} (\bibinfo {year} {2023})}\BibitemShut {NoStop}%
\bibitem [{\citenamefont {Mele}(2024)}]{mele2024introduction}%
  \BibitemOpen
  \bibfield  {author} {\bibinfo {author} {\bibfnamefont {A.~A.}\ \bibnamefont {Mele}},\ }\bibfield  {title} {\bibinfo {title} {Introduction to haar measure tools in quantum information: A beginners tutorial},\ }\href {https://doi.org/10.22331/q-2024-05-08-1340} {\bibfield  {journal} {\bibinfo  {journal} {Quantum}\ }\textbf {\bibinfo {volume} {8}},\ \bibinfo {pages} {1340} (\bibinfo {year} {2024})}\BibitemShut {NoStop}%
\bibitem [{\citenamefont {Collins}\ and\ \citenamefont {Śniady}(2006)}]{CS06}%
  \BibitemOpen
  \bibfield  {author} {\bibinfo {author} {\bibfnamefont {B.}~\bibnamefont {Collins}}\ and\ \bibinfo {author} {\bibfnamefont {P.}~\bibnamefont {Śniady}},\ }\bibfield  {title} {\bibinfo {title} {Integration with respect to the haar measure on unitary, orthogonal and symplectic group},\ }\href {https://doi.org/10.1007/s00220-006-1554-3} {\bibfield  {journal} {\bibinfo  {journal} {Communications in Mathematical Physics}\ }\textbf {\bibinfo {volume} {264}},\ \bibinfo {pages} {773–795} (\bibinfo {year} {2006})}\BibitemShut {NoStop}%
\bibitem [{\citenamefont {Galvin}(2014)}]{galvin2014tutoriallecturesentropycounting}%
  \BibitemOpen
  \bibfield  {author} {\bibinfo {author} {\bibfnamefont {D.}~\bibnamefont {Galvin}},\ }\href {https://arxiv.org/abs/1406.7872} {\bibinfo {title} {Three tutorial lectures on entropy and counting}} (\bibinfo {year} {2014}),\ \Eprint {https://arxiv.org/abs/1406.7872} {arXiv:1406.7872 [math.CO]} \BibitemShut {NoStop}%
\bibitem [{Note3()}]{Note3}%
  \BibitemOpen
  \bibinfo {note} {Indeed, ${\protect \rm TV}(q_{\rho '},q_{\rho })\leqslant \|\rho \otimes \rho -\rho '\otimes \rho '\|_{\protect \rm tr}/2$ since the Bell distribution is the (diagonal) output of a quantum channel acting on two copies of a state, and $\|\rho \otimes \rho -\rho '\otimes \rho '\|_{\protect \rm tr}\leqslant 2\|\rho -\rho '\|_{\protect \rm tr}$}\BibitemShut {NoStop}%
\end{thebibliography}%
\end{document}